\documentclass[journal,twocolumn]{IEEEtran}
\IEEEoverridecommandlockouts
\usepackage{amsmath,amsfonts}
\usepackage{array}
\usepackage{textcomp}
\usepackage{stfloats}
\usepackage{url}
\usepackage{verbatim}
\usepackage{graphicx}
\usepackage{cite}
\usepackage[noend]{algpseudocode}
\usepackage[linesnumbered,ruled]{algorithm2e}
\usepackage{upgreek}
\usepackage{subfigure}
\usepackage{color}
\def\BibTeX{{\rm B\kern-.05em{\sc i\kern-.025em b}\kern-.08em
    T\kern-.1667em\lower.7ex\hbox{E}\kern-.125emX}}
    
\usepackage{amsthm}
\usepackage{amssymb}
\usepackage{makecell}
\usepackage{cleveref}
\usepackage{enumerate}
\newtheorem{theorem}{Theorem}
\newtheorem{definition}{Definition}
\theoremstyle{remark}
\theoremstyle{proposition}
\newtheorem{remark}{Remark}
\newtheorem{proposition}{Proposition}
\theoremstyle{Assumption}

\theoremstyle{Lemma}

\newtheorem{corollary}{\textbf{Corollary}}
\graphicspath{{figures/}}
\usepackage{enumitem}

\newcommand{\myref}[1]{%
    \ifthenelse{\equal{#1}{proof:proposition_1}}{A}{%
    \ifthenelse{\equal{#1}{proof:proposition_2_1}}{B}{%
    \ifthenelse{\equal{#1}{proof:proposition_2_2}}{C}{%
    \ifthenelse{\equal{#1}{proof:proposition_3}}{D}{%
    \ifthenelse{\equal{#1}{proof:proposition_3_3}}{E}{%
    {\textbf{??}}}}}}}%
}

\ifodd 0
\definecolor{darkgreen}{RGB}{0,200,0}
\else

\fi
\begin{document}
\title{
PartialLoading: User Scheduling and Bandwidth Allocation for Parameter-sharing Edge Inference  \\
\thanks{Guanqiao Qu, Qian Chen, Xianhao Chen, and Kaibin Huang are with the Department of Electrical and Electronic Engineering, University of Hong Kong, Hong Kong SAR, China. (e-mail: gqqu@eee.hku.hk; qchen@eee.hku.hk; xchen@eee.hku.hk; huangkb@eee.hku.hk). Yuguang Fang is with the Hong Kong JC STEM Lab of Smart City and the Department of Computer Science, City University of Hong Kong, Hong Kong SAR, China. (e-mail: my.Fang@cityu.edu.hk). The work was supported in part by the Research Grants Council of Hong Kong under Grant 27213824 and CRS HKU702/2. The work of Y. Fang was supported in part by the Hong Kong SAR Government under the Global STEM Professorship. \textit{(Corresponding author: Xianhao Chen.)}

}
}
\author{Guanqiao Qu,~\IEEEmembership{Graduate Student Member,~IEEE}, Qian Chen,~\IEEEmembership{Member,~IEEE}, \\Xianhao Chen,~\IEEEmembership{Member,~IEEE}, Kaibin Huang,~\IEEEmembership{Fellow,~IEEE}, Yuguang Fang,~\IEEEmembership{Fellow,~IEEE}}

\maketitle

\begin{abstract}
By provisioning inference offloading services, edge inference drives the rapid growth of AI applications at network edge. However, how to reduce the inference latency remains a significant challenge. To address this issue, we develop a parameter-sharing AI model loading (PartialLoading) framework for multi-user edge inference, which exploits two key insights: 1) the majority of latency arises from loading AI models into server GPU memory, and 2) different AI models can share a significant number of parameters, for which redundant loading should be avoided. Towards this end, we formulate a joint multi-user scheduling and spectrum bandwidth allocation problem to maximize task throughput by exploiting shared parameter blocks across models. The intuition is to judiciously \textit{schedule} user requests to reuse the shared parameter blocks between consecutively loaded models, thereby reducing model loading time substantially. To facilitate solution finding, we decouple the problem into two sub-problems, i.e., user scheduling and bandwidth allocation, showing that solving them sequentially leads to the solution to the original problem. Due to the NP-hardness of the problem, we first study an important special case called the ``backbone-sharing'' case, and design a dynamic programming-based algorithm to obtain the \textit{optimal} solution in polynomial time. For the general case, we propose a greedy heuristic to obtain the sub-optimal solution efficiently. Simulation results demonstrate that the proposed framework significantly improves task throughput under deadline constraints compared with user scheduling without exploiting parameter sharing.
\end{abstract}
\begin{IEEEkeywords}
Edge AI, edge inference, 6G, user scheduling, bandwidth allocation, task throughput maximization.
\end{IEEEkeywords}
\section{Introduction}
The sixth-generation (6G) mobile networks are fundamentally reshaping the network paradigm through ``Integrated Artificial Intelligence (AI) and Communications"~\cite{ITU-R-M.2160-0}. With this integration, 6G is envisioned to support ubiquitous AI applications at the network edge. A pivotal usage scenario in this vision is \textit{edge inference}, where end devices offload compute-intensive AI inference tasks to nearby wireless edge servers (e.g., base stations), thereby significantly alleviating the computing burden on end devices~\cite{qu2024mobile}. Unlike cloud-based inference, by processing data in proximity to end users, edge inference reduces data transmission latency~\cite{zhang2019mobile,chen2024space}, thus enabling a wide range of real-time mobile AI applications~\cite{lim2022realizing,dai2019industrial,amin2020edge}.

Despite having low transmission latency, edge inference can still face significant end-to-end (E2E) latency due to high inference latency. Within overall inference latency, \textit{model loading/swapping} is often the dominating contributor. For instance, loading an AI model into the GPU memory can account for an average of 88.94\% of the total inference latency (as shown in Fig. \ref{fig:intro_comp}). For large models with billions of parameters, e.g., GPT-2, this overhead can become more severe, reach 96.3\% of the inference latency~\cite{han2024hermes}. This excessive overhead stems from the fundamental separation of processing and memory units, known as the ``\textit{Von Neumann bottleneck}". Although this bottleneck affects both cloud and edge deployments, it is more pronounced at the edge where edge servers possess more limited computing capabilities. This limitation can hinder timely responses for real-time AI applications, such as mobile Augmented Reality (AR), which demands 7-15 ms E2E latency~\cite{3gpp.22.874}. While various resource allocation schemes have been developed to jointly optimize communication and computing resources for real-time edge inference~\cite{9134426,9311935}, most of these efforts have overlooked the critical model loading bottleneck.

\begin{figure}[t]
\centering
\includegraphics[width=0.27\textwidth]{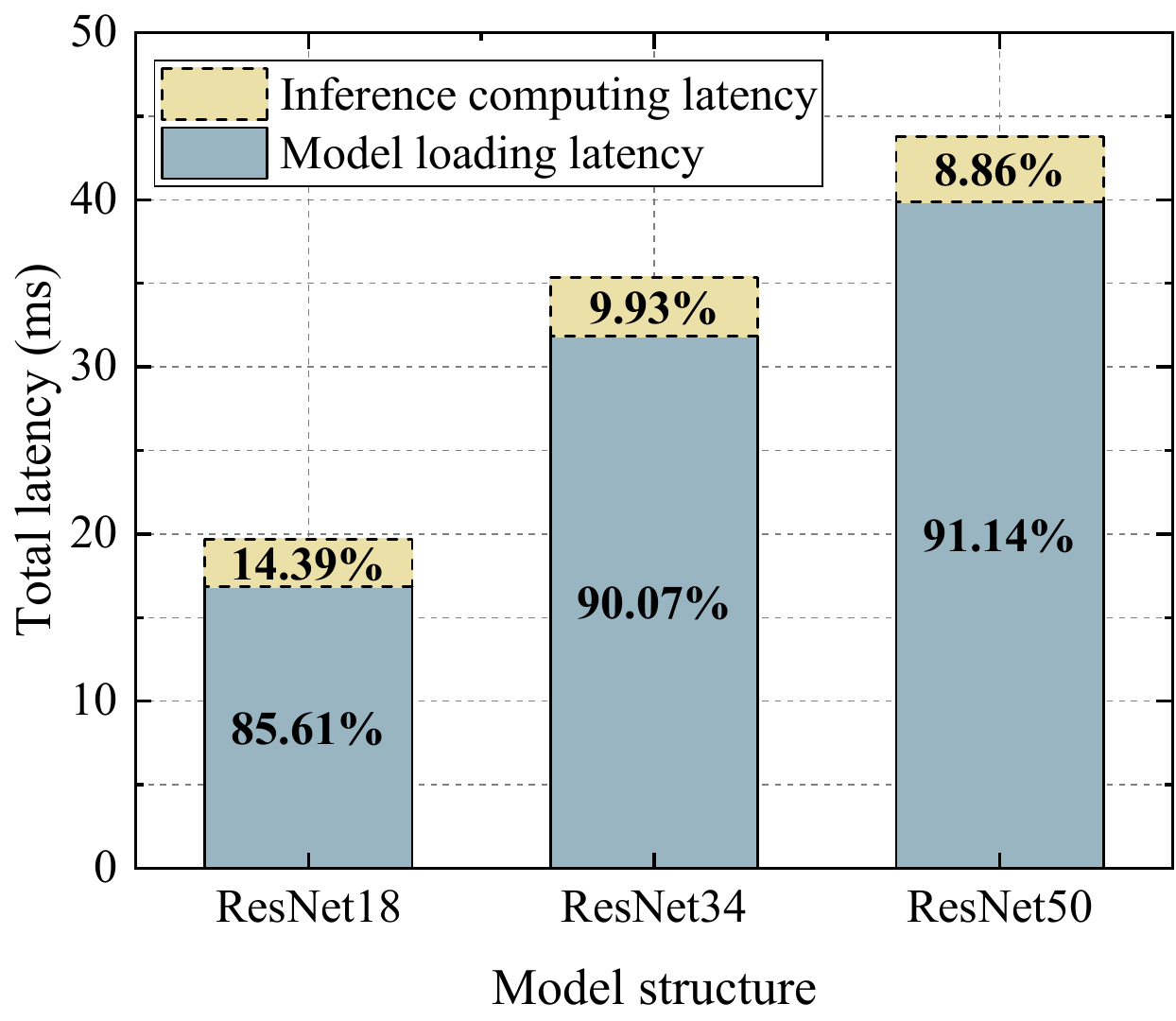}
\vspace{-0.25cm}
\caption{Total latency across various model structures. Model loading refers to the process of loading an AI model to GPU memory, while 
inference computing includes data tensorization and batching, moving data to GPU memory, and forward propagation. The models are from the ResNet family \cite{he2016deep}, with a batch size of 32, evaluated on the CIFAR10 dataset \cite{krizhevsky2009learning}. Inference is executed on a Linux server equipped with an Intel Core i9-13900K CPU, a GeForce RTX 4090 GPU with 24 GB GPU memory, a Toshiba 8 TB SATA3 HDD, and two Kingston 32 GB DDR5 RAM modules.}
\label{fig:intro_comp}
\end{figure}


To overcome the aforementioned challenges, our key idea is to leverage \textit{parameter sharing} across AI models to accelerate the edge inference process. Parameter sharing is a common property in AI models like convolutional neural networks (CNNs) and large language models (LLMs). Specifically, with the emergence of layer-freezing or parameter-efficient fine-tuning (PEFT) techniques, a substantial proportion of parameters of the pre-trained model can be reused across downstream models, with only a small subset of parameters being updated\footnote{Parameter sharing during model fine-tuning effectively reduces computing workload and prevents overfitting while incurring minimal degradation on accuracy. For example, Fig. \ref{fig:intro_50_acc} shows that two downstream models incur no accuracy degradation in the corresponding tasks compared with full-parameter fine-tuning (0 frozen layers), even when up to 93\% and 89\% of the bottom layers are shared from the pre-trained model.}~\cite{tenser2023,zhuang2020comprehensive,qu2024trimcaching,Guo_2019_CVPR,10.5555/2969033.2969197,hu2021lora}. Nevertheless, while parameter sharing has been widely exploited to improve training efficiency, its potential for accelerating \textit{inference} has been largely untapped. An open question remains: \textit{Given a set of models with shared parameters on an edge server, how can this fact be optimally leveraged to enhance inference efficiency by reducing model loading time?} 


\begin{figure}[t]
\centering
\includegraphics[width=0.27\textwidth]{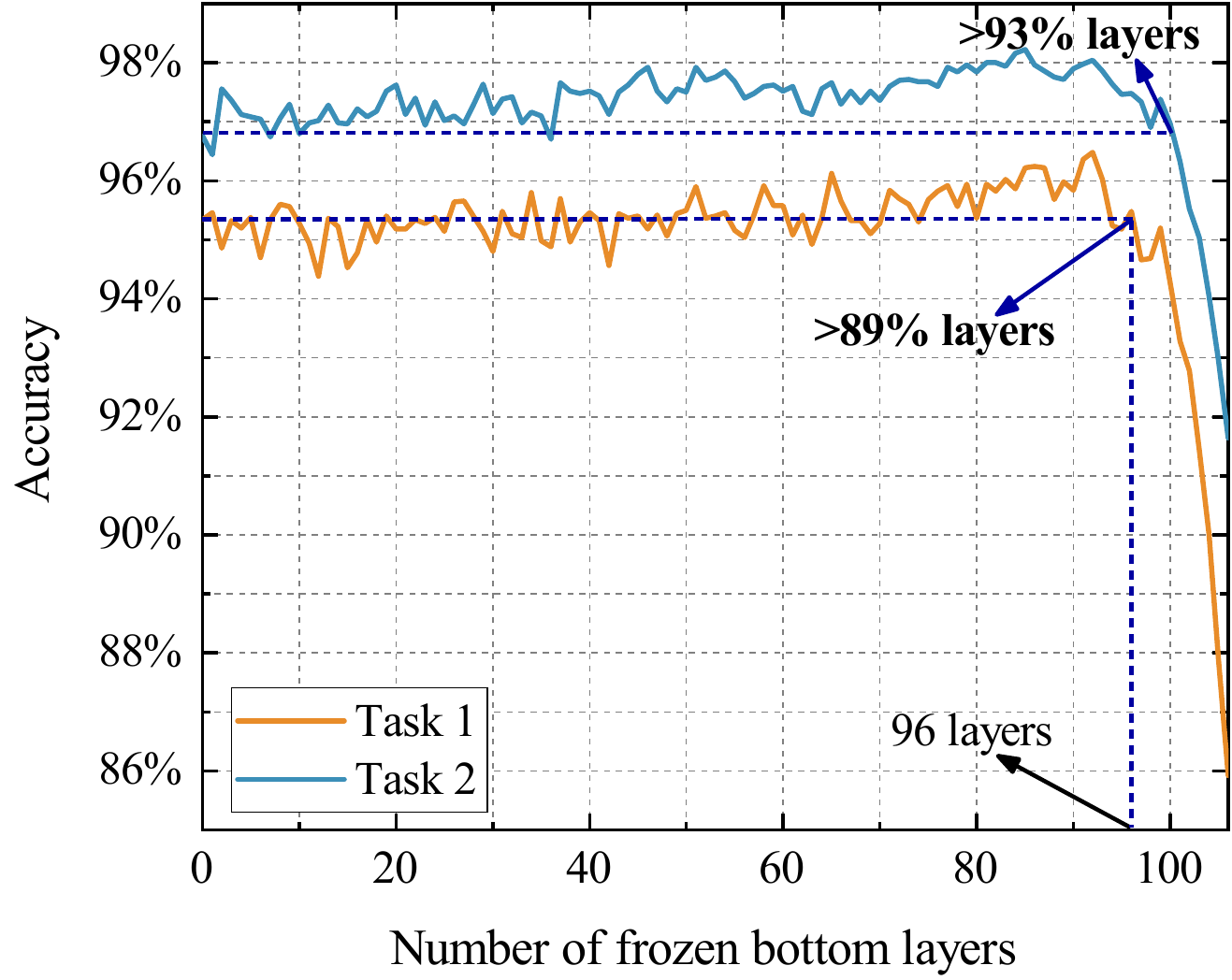}
\vspace{-0.25cm}
\caption{Inference accuracy vs. the number of frozen bottom layers in fine-tuned ResNet-50 models~\cite{he2016deep}, where the ResNet-50 is first pre-trained on CIFAR-100~\cite{krizhevsky2009learning} and then fine-tuned on CIFAR-10~\cite{krizhevsky2009learning}. Tasks 1 and 2, respectively, correspond to the classification of labels 0-4 and 5-9 in CIFAR10. This example shows that different downstream models can share a significant proportion of layers without performance degradation compared with full-parameter fine-tuning (0 frozen layers).}
\vspace{-10pt}
\label{fig:intro_50_acc}
\end{figure}

In this paper, we consider an edge inference scenario in which a resource-constrained edge server provides inference services for users using multiple models with shared parameters. Moreover, user requests for the same model can be processed in batches to improve inference throughput. Our key insight is that, when models in two consecutively executed batches share parameters, the shared parameters can remain in GPU memory and be directly reused by the subsequent batch to avoid redundant parameter loading. 
However, developing an optimal user scheduling policy for the system presents a complex challenge. Simply maximizing parameter overlapping alone is \textit{insufficient}, as the E2E latency of a batch is also constrained by user uplink communication latency: a straggler in a batch can prolong the latency of the entire batch and subsequent batches. Consequently, the scheduling policy must balance two potentially conflicting goals: (i) scheduling requests for the same model together or parameter-sharing models in consecutive batches to minimize model loading latency, and (ii) scheduling users by their uplink latency to mitigate the E2E delay caused by stragglers in a batch. This gives rise to a tradeoff between model loading latency and uplink communication latency.

To tackle the above challenges, this paper addresses the joint multi-user scheduling and spectrum bandwidth allocation problem for parameter-sharing edge inference. By jointly considering model loading and user uplink communication latency, we design scheduling strategies to schedule user inference requests for AI models to maximize inference efficiency. To the best of our knowledge, this is the first work to optimize edge inference in wireless networks by exploiting parameter sharing among AI models. Our main contributions are summarized as follows.
\begin{enumerate}[label=\arabic*)]
\item We propose a framework leveraging parameter sharing across AI models to enable partial model loading for wireless edge inference. 
We formulate a joint user scheduling and spectrum bandwidth allocation problem in multi-user parameter-sharing edge inference systems to maximize the task throughput under latency, communication, and computing resource constraints. 

\item To facilitate solution finding, we effectively decouple the problem into two sub-problems, i.e., user scheduling and spectrum bandwidth allocation, demonstrating that solving the original problem is equivalent to solving these two sub-problems sequentially. Then, we derive the optimal closed-form expression for bandwidth allocation under given user scheduling solutions. 


\item Given the NP-hardness of the problem, we first investigate an important special case of the original problem, where AI models within clusters share backbones. For this backbone-sharing case, we first obtain the optimal model loading order, based on which a dynamic programming (DP)-based algorithm is proposed to obtain the optimal solution with polynomial-time complexity.

\item For the general case, where shared parameter blocks appear at arbitrary positions across models, we design a greedy heuristic algorithm to find a sub-optimal solution. Both of our algorithms considerably outperform the traditional user scheduling strategy without exploiting the parameter-sharing property during model loading. 
\end{enumerate} 

The rest of this paper is organized as follows. Section II reviews the related work. Section III introduces the system model and the proposed framework. Section IV formulates the task throughput maximization problem, presents the equivalent problem decoupling, derives the optimal bandwidth allocation policy, and establishes the necessary conditions for optimal user scheduling. Section V investigates a special case of the problem, while Section VI addresses the general case. Section VII presents simulation results. Finally, Section VIII concludes the paper.


\section{Related Work}
In multi-user edge inference, enhancing task throughput under latency constraints is a challenging issue. This difficulty primarily comes from the conflict between stringent task latency requirements and limited communication and computing resources on an edge server. To improve task throughput, a common approach is to leverage task batching, i.e., scheduling the requests for the same AI model into one inference batch~\cite{liu2023resource,shi2022multiuser,cang2024joint,zhang2024edge,ma2024efficient,yang2019computation,she2024dynamic,yao2022loading}. By batching them into an AI (sub)model, an edge server processes inference computing in parallel based on the parallel computing of GPU. Building on this approach, some research works optimize data selection for inference batching. However, these works focus on scheduling users into a single inference batch~
\cite{liu2023resource,zhang2024edge,yang2019computation,she2024dynamic,yao2022loading}, limiting their applicability in the real world. While some studies have also addressed the user scheduling problems for multiple batches, they make the simplified assumption that users request the same AI model~\cite{shi2022multiuser,cang2024joint,ma2024efficient}. In practice, an edge server typically hosts diverse AI models for its subscribers for different applications.


As alluded to earlier, loading AI models into the server memory is a major bottleneck in the inference process. 
In response to this need, Ogden et al. \cite{ogden2021many} design a model eviction policy that optimizes the models cached in memory for new inference requests by minimizing the cache miss penalty, where the penalty occurs when a requested model is not cached. However, it does not leverage parameter sharing across AI models to reduce model loading time. In \cite{padmanabhan2022gemel}, Padmanabhan et al. propose a loading strategy for edge inference with parameter-sharing models, where the next model is scheduled to maximize the overlapping of shared layers with the previous one, thereby reducing model loading time. Nevertheless, this heuristic fails to account for the data uploading and computing time and does not provide a provable performance guarantee. Beyond these works based on traditional architectures, in-memory computing~\cite{kwon2022yono} has emerged as an alternative architecture to process data directly in memory. However, in-memory computing faces incompatibility with traditional commercial systems that still dominate today's computing industries, limiting its feasibility for widespread deployment in practice. Overall, these works overlook the communication constraints of edge servers, limiting the practicality of their approaches.

The key idea of our proposed framework is strategically scheduling user requests to reuse the shared parameter blocks between consecutively loaded AI models. However, while the performance gain is clear, solving the optimization problem is highly challenging. Specifically, assuming AI models are independent or there is only one AI model, the scheduling problems in edge inference can be typically mapped to knapsack problems and solved with standard DP algorithms~\cite{liu2023resource,shi2022multiuser}. Unfortunately, when considering parameter sharing across models, the interdependencies among models significantly complicate user scheduling, causing standard DP procedures to fail. Moreover, incorporating bandwidth allocation further increases the complexity of the problem. To the best of our knowledge, this is the first paper to address the user scheduling problem in multi-user edge inference by exploiting parameter sharing across models. 


\section{Framework}
\subsection{Network Scenario}
As illustrated in Fig. \ref{fig:system_model}, we consider a single-edge multi-user scenario in wireless edge networks. A set of users, $\mathcal{K}=\left\{1,2,\dots, K\right\}$, are associated with a wireless edge server (e.g., a server colocated with a base station). Without loss of generality, we assume each user generates one inference task requesting a specific AI model. The set of models is denoted as $\mathcal{I}=\left\{1,2,\dots, I\right\}$, which are hosted on the edge server, and each model may be requested by multiple users. Each user uploads input data of the task\footnote{To address privacy concerns, users can upload a feature vector extracted from raw data using a lightweight on-device model or the shallow layers of the requested model, which is equivalent to user data uploading in this paper.} to the edge server for edge inference, after which the server returns the inference results to the user. Moreover, the time horizon is divided into $T$ consecutive time slots with equal duration $\Delta\tau$, indexed by $\tau\in\mathcal{T}=\left\{1,2,\dots\right\}$. 

As a remark, unlike prior works that predominantly focus on an oversimplified \textit{homogeneous} task modeling where all users request the same AI model for inference~\cite{liu2023resource,shi2022multiuser}, this paper addresses a more practical \textit{heterogeneous} task modeling by considering multiple AI models/tasks hosted on an edge server, which is prerequisite for applying our scheduling policy. 

\begin{figure}[t]
\centering
\includegraphics[width=0.35\textwidth]{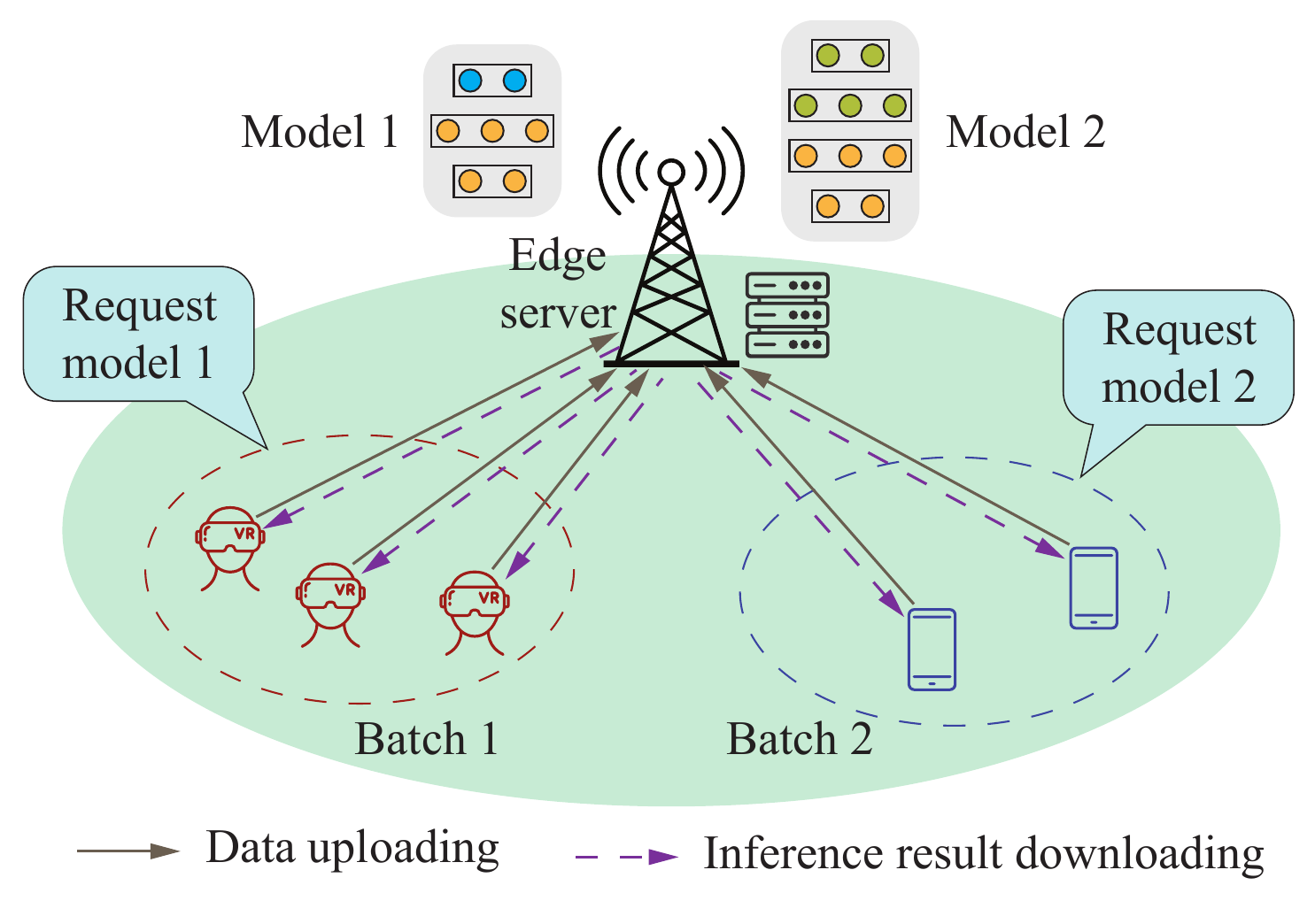}
\vspace{-0.25cm}
\caption{Our considered single-edge multi-user scenario, where parameter blocks can be shared among AI models hosted on the server. In this figure, the first two layers (in orange) are shared between model 1 and model 2.}
\label{fig:system_model}
\vspace{-10pt}
\end{figure}

\subsection{Parameter-sharing AI Model Library}
We consider the parameter-sharing in the AI model library $\mathcal{I}$. The set $\mathcal{J}=\left\{1,2,\dots,J\right\}$ represents the parameter blocks of the models in $\mathcal{I}$. Given that AI models may contain billions of parameters, a parameter block refers to a set of model parameters, such as a (sub-)layer, a model block, or even a backbone (e.g., the bottom layers of a model or a frozen pre-trained model), to reduce the complexity of the problem. This applies to a broad range of AI models like CNNs, Transformers, and LLMs. The definition of a shared parameter block is formally presented below.
\begin{definition}
    \textbf{Shared parameter block}: A parameter block is referred to as a shared parameter block if it is contained in more than one AI model within a considered set of models.
\end{definition}
Moreover, when a parameter block corresponds to a backbone or to a layer, a shared parameter block is also called a \textit{shared backbone} or a \textit{shared layer} in this paper.


\subsection{Inference Batching with Partial Model Loading}
To enhance task throughput, the server processes tasks in batches. Following prior work~\cite{liu2023resource,shi2022multiuser,ren2019collaborative,li2021adaptive}, we assume that all task requests arrive at the edge server at time 0\footnote{Although users can send their task requests at random time to the edge server, the edge server can initiate user data uploading after receiving a number of requests during a time duration to improve inference efficiency, which is equivalently modeled by assuming all task requests arrive at time 0~\cite{10682995,10549973,9708995}. Alternatively, our work can also apply to the case where the arrival time of tasks is predictable at the beginning.}, and multiple tasks requesting the same model can be batched together to leverage the parallel processing capabilities of GPUs. The batch sequence is denoted by $n\in\mathcal{N}=\left\{1,2,\dots, N\right\}$, where the edge server processes batches sequentially in ascending order of $n$. In batch $n$, only users requesting the same AI model can be scheduled. Let a binary variable $x_{n,k}$ be the decision variable for user scheduling, with $x_{n,k}=1$ indicating that the server processes user $k$'s task in batch $n$. As illustrated in Fig. \ref{fig:workflow}, each inference batch in the proposed parameter-sharing AI model loading (PartialLoading) framework consists of the following three steps.

    (1) \textbf{Data uploading}: Scheduled users upload raw data or features to the edge server. The uploaded data is cached in the system memory of the edge server after reception. 
    
    (2) {\textbf{Partial model loading}}: Once all data from the scheduled users is uploaded, the edge server loads the specified AI model from the disk to system memory and then to GPU memory. To reduce latency, the shared parameter blocks between two consecutively loaded models in the current and the previous batch are reused without reloading.
    
    (3) \textbf{Inference computing}: Inference computing comprises three steps. (3.1) \textit{Data tensorization and batching}: The edge server converts the uploaded raw data into tensors, which can be skipped if users upload preprocessed features directly. Then, the edge server assembles the tensorized data into a single batch. (3.2) \textit{Moving data to GPU}: The edge server moves the batched data from the system memory to the GPU memory for further processing. (3.3) \textit{Forward propagation}: The edge server feeds the batched data into the AI model in the GPU memory, performs forward propagation, produces inference results, and sends the results back to users.  
\begin{figure}[t]
\centering
\includegraphics[width=0.48\textwidth]{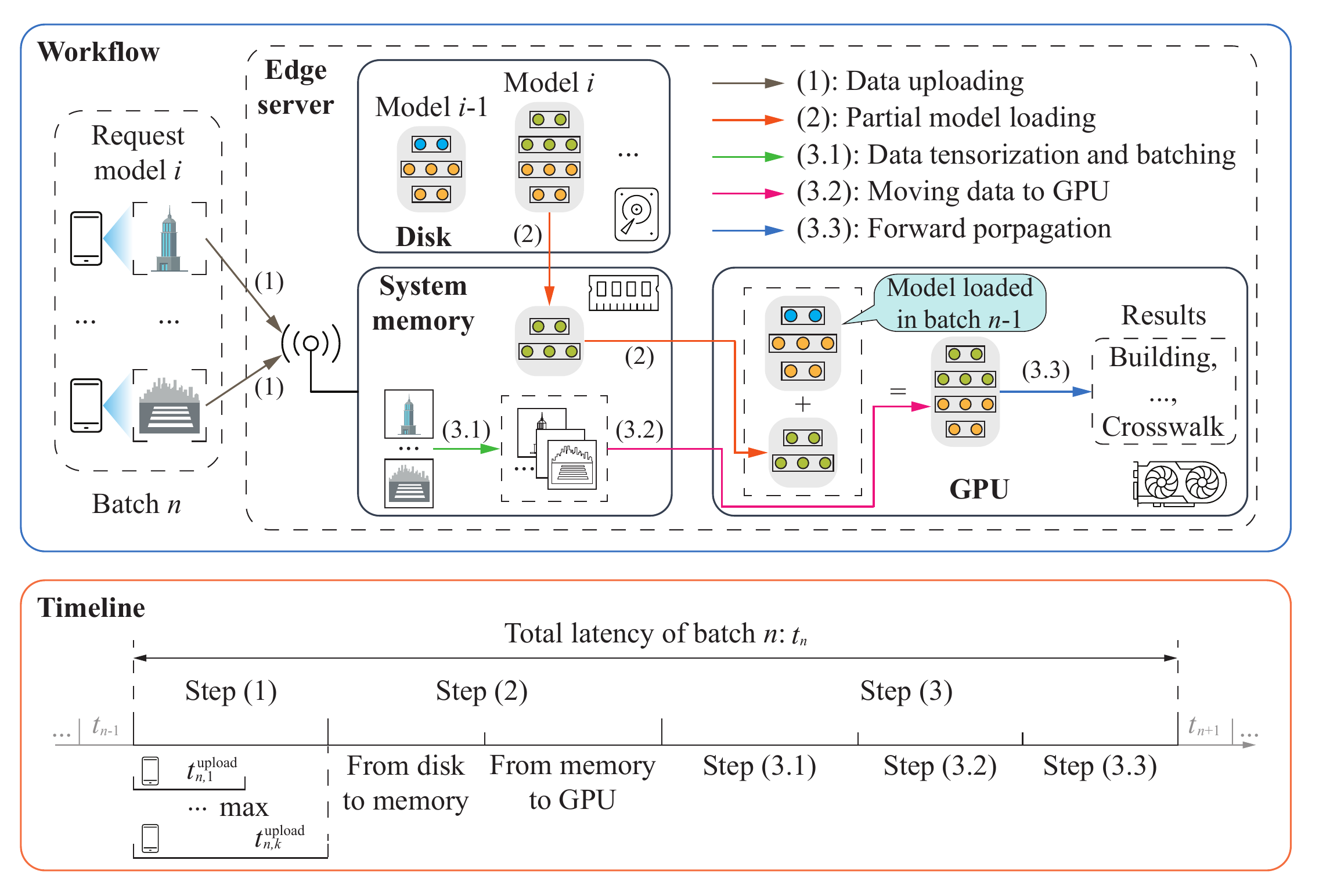}
\vspace{-0.25cm}
\caption{Workflow and timeline of an inference batch. We use object classification/detection applications as an illustrative example. Mobile devices offload the captured images to the edge server and request model $i$ in batch $n$. Since model $i$ shares the first two layers (in orange) with model $i-1$, which has been loaded in the previous batch (batch $n-1$), the edge server only loads the last two layers (in green) of model $i$ into the GPU memory for inference. The main goal of this paper is to schedule users into a sequence of batches to optimally leverage the parameters shared across models, thereby enhancing task throughput under latency constraints.}
\label{fig:workflow}
\vspace{-10pt}
\end{figure}

Due to the small data size of the inference results and the high transmit power of the edge server, the time of downloading inference results is ignored \cite{cang2024joint}. As a result, the total latency for processing batch $n$ is expressed as
\begin{equation}\label{eq_t_n}
    t_{n} = t_{n}^{\text{up}} + t_{n}^{\text{load}} + t_{n}^{\text{comp}},
\end{equation}
where $t_{n}^{\text{up}}$ represents the data uploading time of scheduled users, $t_{n}^{\text{load}}$ denotes the time for partial model loading, and $t_{n}^{\text{comp}}$ is the inference computing time. The detailed timeline of an inference batch is illustrated in Fig. \ref{fig:workflow}, and further details will be provided in the following subsections.

User $k$ can be scheduled and served in batch $n$ if it can obtain the inference result within the E2E latency requirement, satisfying
\begin{equation}\label{const_4}
    x_{n,k}\sum\limits_{n'=1}^{n}t_{n'} \le T\Delta\tau,\ \forall n\in{\mathcal{N}},\ \forall k\in{\mathcal{K}},
\end{equation}
where $\bar{T}=T\Delta\tau$ denotes the E2E latency requirement.
Additionally, since user $k$ can be scheduled in at most one batch, it follows that
\begin{equation}\label{const_1}
    \sum\limits_{n\in\mathcal{N}}x_{n,k}\le 1, \ \forall k \in\mathcal{K}.
\end{equation}

\subsection{Uplink Communication Time}
We consider an FDMA uplink channel between the users and the edge server. In each batch, the total bandwidth $B$ of the edge server is divided into multiple sub-channels, each allocated to a single user. The bandwidth allocated to user $k$ in batch $n$ is denoted by 
\begin{equation}\label{eq_bw}
B_{n,k}=y_{n,k}B,
\end{equation}
where $y_{n,k}\in\left[0,1\right]$ is the ratio of the allocated bandwidth, satisfying
\begin{equation}\label{const_6}
    \sum\limits_{k\in\mathcal{K}}y_{n,k} \le 1, \ \forall n\in\mathcal{N}.
\end{equation}

Moreover, the expected data uploading rate between user $k$ and the edge server is given by $
\bar{C}_{n,k} = B_{n,k}\bar{R}_{k}$, where $\bar{R}_{k}={\rm{log}}_2\left(1+\frac{P_{k}d_{k}^{-\alpha}}{N_{0}}\right)$, $P_{k}$ is the spectral density of the transmit power of user $k$, $d_{k}$ is the distance between user $k$ and the server, $\alpha$ is the path loss
factor, and $N_0$ is the spectral density of the additive white Gaussian noise. Thus, the data uploading time of user $k$ in batch $n$ is
\begin{equation}\label{eq:latency_uploading}
    t_{n,k}^{\text{up}} = 
    \begin{cases}
                    \begin{aligned}
				&\frac{x_{n,k}D_{k}}{y_{n,k}B\bar{R}_{k}},\ \text{if }x_{n,k}=1,\\
				&0,\ \text{if }x_{n,k}=0,
                    \end{aligned}
		\end{cases}
\end{equation}
where $D_{k}$ is the size of the data uploaded by user $k$. Furthermore, since the partial model loading step begins only after all scheduled users in batch $n$ finish uploading their data, $t_{n}^{\text{up}}$ in \eqref{eq_t_n} is expressed as
\begin{equation}  
    t_{n}^{\text{up}} = \mathop{\max}\limits_{k\in\mathcal{K}}\left\{t_{n,k}^{\text{up}}\right\}.
\end{equation}

\subsection{Partial Model Loading Time}
In the partial model loading step of batch $n$, an edge server only loads the parameter blocks not shared between models in batches $n-1$ and $n$. Therefore, $t_{n}^{\text{load}}$ in \eqref{eq_t_n} is
\begin{equation}\label{eq_partial_load}
\begin{aligned}
    t_{n}^{\text{load}} = V_{1}\left(\bar{S}_{n}\right)
    +V_{2}\left(\bar{S}_{n}\right),
\end{aligned}
\end{equation}
where $\bar{S}_{n} =\sum\limits_{j\in\mathcal{J}}S_{j}\rho_{n,j}\left(\rho_{n,j}-\rho_{n-1,j}\right)$ is the size of parameter blocks to be loaded in batch $n$, and $S_{j}$ is the size of parameter block $j$. Here, $\rho_{n,j}=1-\prod\limits_{i\in\mathcal{I}_{j}}\prod\limits_{k\in\mathcal{K}_{i}}\left(1-x_{n,k}\right)$, where $\mathcal{I}_{j}$ is the set of models containing parameter block $j$, $\mathcal{K}_{i}$ is the set of users requesting model $i$. $\rho_{n,j}=1$ indicates that parameter block $j$ is included in the request model in batch $n$, where $\rho_{0,j}=0$. Thus, $\rho_{n,j}\left(\rho_{n,j}-\rho_{n-1,j}\right)=1$ implies that parameter block $j$ has not been loaded in batch $n-1$ and needs to be loaded in batch $n$. Moreover, functions $V_{1}\left(\cdot\right)$ and $V_{2}\left(\cdot\right)$ denote the loading time of data from the disk to system memory and from system memory to GPU memory, respectively, which are monotone increasing functions with data size\footnote{In general, $V_{1}\left(\cdot\right)$ depends on the bandwidth between the disk and system memory, as well as on the page cache of the system memory \cite{nvidiagpu,linux}. $V_{2}\left(\cdot\right)$ can be approximated as the ratio of the data size to the bandwidth between the system memory and GPU memory~\cite{xu2022igniter}.}.

\subsection{Inference Computing Time}
The inference computing time increases linearly as the batch size grows \cite{huang2018multi,wang2024joint,wu2023graft}. 
Given $x_{n,k}$, $t_{n}^{\text{comp}}$ in \eqref{eq_t_n} is
\begin{equation}
    t_{n}^{\text{comp}} = \sum\limits_{i\in\mathcal{I}}\left[1-\prod\limits_{k\in\mathcal{K}_{i}}\left(1-x_{n,k}\right)\right]\left(\mu_{i}\sum\limits_{k\in\mathcal{K}_{i}}x_{n,k}+\beta_{i}\right),
\end{equation}
where $1-\prod\limits_{k\in\mathcal{K}_{i}}\left(1-x_{n,k}\right) = 1$ indicates model $i$ is loaded in batch $n$, $\mu_{i}$ and $\beta_{i}$ are constants related to the model structure and GPU performance when performing inference with model $i$, and $\sum\limits_{k\in\mathcal{K}_{i}}x_{n,k}$ is the batch size of batch $n$ for model $i$.

\section{Task Throughput Maximization Problem}
\subsection{Problem Formulation} 
Our framework aims to maximize task throughput by addressing the joint user scheduling and bandwidth allocation problem within the edge server's GPU memory capacity, communication constraints, and users' task latency requirements. The problem formulation is given as follows. 

\begin{subequations}
	\begin{equation}
		{\mathcal{P}1}:\ \mathop{\max}\limits_{{\bf{X}},{\bf{Y}}}\ U\left({\bf{X}}\right) = \sum\limits_{n\in\mathcal{N}}\sum\limits_{k\in\mathcal{K}}x_{n,k}
	\end{equation}	
        \begin{equation}
		{\rm{s.t.}} \ \eqref{const_4}, \eqref{const_1},\eqref{const_6}, 
	\end{equation}	
	\begin{equation}\label{const_2}
		\sum\limits_{i\in\mathcal{I}}A_{i}\left(\sum\limits_{k\in\mathcal{K}_{i}}x_{n,k}\right)\le Q,\ \forall n\in{\mathcal{N}},
	\end{equation}	
        \begin{equation}\label{const_3}
        \sum\limits_{i\in\mathcal{I}}\left[1-\prod\limits_{k\in\mathcal{K}_{i}}\left(1-x_{n,k}\right)\right]\le 1,\ \forall n\in{\mathcal{N}},
	\end{equation}
        \begin{equation}\label{const_7}
		y_{n,k}\le x_{n,k},\ \forall n\in{\mathcal{N}},\ \forall k\in{\mathcal{K}},
	\end{equation}	
	\begin{equation}\label{const_5}
		x_{n,k}\in\left\{0,1\right\},\ \forall n\in{\mathcal{N}},\ \forall k\in{\mathcal{K}},
	\end{equation}	
        \begin{equation}\label{const_8}
		y_{n,k}\in\left[0,1\right],\ \forall n\in{\mathcal{N}},\ \forall k\in{\mathcal{K}},
	\end{equation}	
\end{subequations}
where $x_{n,k}\in{\bf{X}}$ and $y_{n,k}\in{\bf{Y}}$ denote the decision variables for user scheduling and spectrum bandwidth allocation, respectively. In \eqref{const_2},
$Q$ is the edge server's maximum available GPU memory, and function $A_{i}\left(\sum\limits_{k\in\mathcal{K}_{i}}x_{n,k}\right)$ is the peak GPU memory required for forward propagation with model $i$ under batch size $\sum\limits_{k\in\mathcal{K}_{i}}x_{n,k}$, which is a monotone increasing function~\cite{padmanabhan2022gemel}.  \eqref{const_3} ensures at most one model is loaded into GPU memory in batch $n$. \eqref{const_7} guarantees user $k$ is not allocated bandwidth unless it is scheduled in batch $n$. 

\subsection{Equivalent Problem Decoupling}
$\mathcal{P}1$ is a mixed-integer nonlinear programming problem coupling ${\bf{X}}$ and ${\bf{Y}}$. To simplify the solution process, we derive the optimal closed-form expression for spectrum bandwidth allocation under ${\bf{X}}$. The details are as follows.
\begin{proposition}\label{proposition_1}
    Given any user scheduling decision in batch $n$, the corresponding optimal spectrum bandwidth allocation in batch $n$ can be obtained by minimizing $t_{n}^{\text{up}}$. For given ${\bf{X}}$, the minimum value of $t_{n}^{\text{up}}$ is 
    \begin{equation}\label{optimal_uploading}
        t_{n,\min}^{\text{up}} = \sum\limits_{k\in\mathcal{K}}\frac{x_{n,k}D_{k}}{B\bar{R}_{k}},
    \end{equation}
    and the corresponding optimal spectrum bandwidth allocation under ${\bf{X}}$ is 
    \begin{equation}\label{optimal_bandwidth}
        y_{n,k} = 
        \begin{cases}
                    \begin{aligned}
				&\frac{x_{n,k}D_{k}}{B\bar{R}_{k}\sum\limits_{k'\in\mathcal{K}}\frac{x_{n,k'}D_{k'}}{B\bar{R}_{k'}}},\ \text{if }x_{n,k} = 1,\\
				&0,\ \text{if }x_{n,k}=0.
                    \end{aligned}
		\end{cases}
    \end{equation}
\end{proposition}
\begin{proof}
    The proof is provided in Appendix \myref{proof:proposition_1}.
\end{proof}

Building on Proposition \ref{proposition_1}, we can decouple $\mathcal{P}1$ without compromising the optimality, as stated in the following theorem.
\begin{theorem}\label{theorem_1}
    Solving the original problem $\mathcal{P}1$ is equivalent to first solving the below sub-problem $\mathcal{P}2$ and then determining the optimal spectrum bandwidth allocation using \eqref{optimal_bandwidth} based on the obtained user scheduling.
\end{theorem}
\begin{subequations}
	\begin{equation}
		{\mathcal{P}2}:\ \mathop{\max}\limits_{{\bf{X}}}\ U\left({\bf{X}}\right)
	\end{equation}	
        \begin{equation}
		{\rm{s.t.}} \ \eqref{const_1},\eqref{const_2},\eqref{const_3},\eqref{const_5},
	\end{equation}	
        \begin{equation}\label{const_9}
		x_{n,k}\sum\limits_{n'=1}^{n}\left(t_{n',\min}^{\text{up}} + t_{n'}^{\text{load}} + t_{n'}^{\text{comp}}\right) \le T\Delta\tau,\ \forall n\in{\mathcal{N}},\ \forall k\in{\mathcal{K}}.
	\end{equation}	
\end{subequations}
\begin{proof}
    From Proposition \ref{proposition_1}, for any given $\bf{X}$, the optimal spectrum bandwidth allocation from \eqref{optimal_bandwidth} results in the minimum data uploading time in \eqref{optimal_uploading}, which is independent of $\bf{Y}$. Therefore, after substituting \eqref{optimal_uploading} into $\mathcal{P}1$ and eliminating the constraints on $\bf{Y}$, we obtain $\mathcal{P}2$ that only depends on $\bf{X}$. Solving $\mathcal{P}2$ obviously yields the optimal $U\left({\bf{X}}\right)$. Therefore, solving the user scheduling decision from $\mathcal{P}2$, followed by determining the spectrum bandwidth allocation from \eqref{optimal_bandwidth}, preserves the optimality of the solution to $\mathcal{P}2$. This completes the proof. 
\end{proof}

Thereafter, we focus on solving $\mathcal{P}2$ in the rest of this paper.

\subsection{Problem Mapping}
$\mathcal{P}2$ is a combinatorial optimization problem, as it involves selecting a user-scheduling solution across batches from a discrete set of feasible user-to-batch assignments under constraints. To facilitate analysis, we equivalently reformulate $\mathcal{P}2$ into $\mathcal{P}3$ with the objective of minimizing the number of unserved users as follows.
\begin{subequations}
	\begin{equation}
		{\mathcal{P}3}:\ \mathop{\min}\limits_{{\bf{X}}}\ K-U\left({\bf{X}}\right)
	\end{equation}	
        \begin{equation}
		{\rm{s.t.}} \ \eqref{const_1},\eqref{const_2},\eqref{const_3},\eqref{const_5},\eqref{const_9}.
	\end{equation}	
\end{subequations}

$\mathcal{P}3$ is NP-hard since it can be mapped to a single machine batch scheduling problem to minimize the number of late jobs under latency constraints with a set-up time \cite{allahverdi2008survey}. 
Specifically, in the mapped problem, selected jobs are grouped into batches to be processed sequentially on a single machine, with a set-up time before processing at the beginning of each batch. Similarly, in $\mathcal{P}3$, scheduled users are batched and served by the edge server, with $t_{n,\min}^{\text{up}}+t_{n}^{\text{load}}$ being the set-up time. In fact, even in a simplified case of $\mathcal{P}3$, where $t_{n,\min}^{\text{up}}+t_{n}^{\text{load}}$ is constant and does not vary with $n$, the resultant problem has been proven to be NP-hard \cite{brucker1996single}. Therefore, the general form of $\mathcal{P}3$ and $\mathcal{P}2$ are NP-hard.

\subsection{Necessary Conditions of the Optimal Solution to $\mathcal{P}2$}
To narrow down the feasible user-scheduling solution space of $\mathcal{P}2$, we analyze the scheduling rules for users requesting the same model. The following theorem provides the necessary conditions for the optimal user scheduling ${\bf{X}}^{*}$ in $\mathcal{P}2$ within the batches of each model.

\begin{theorem}\label{proposition_2}
    ${\bf{X}}^{*}$ satisfies the following conditions.
    \begin{enumerate}[label=2.\arabic*] 
        \item\label{proposition_2_1} Users requesting the same AI model must be scheduled in the same or consecutive batches. Moreover, for a given AI model $i$, users in $\mathcal{K}_{i}$ must be scheduled in ascending order of $p_{k}=\frac{D_{k}}{\bar{R}_{k}}$.
        \item\label{proposition_2_2} For consecutive batches requiring the same model, all batches except the last one should admit the maximum number of users within the GPU memory constraint.  %
    \end{enumerate}
\end{theorem}
\begin{proof}
    The proofs of Theorems \ref{proposition_2_1} and \ref{proposition_2_2} are presented in Appendices \myref{proof:proposition_2_1} and \myref{proof:proposition_2_2}, respectively.
\end{proof}
\begin{remark}
    Theorem \ref{proposition_2_1} establishes a user scheduling rule that ensures users with the shortest uploading time given unit spectrum bandwidth allocation are scheduled first within the batches of any given AI model, thereby significantly simplifying the problem. 
\end{remark}

For ease of presentation, we reorder the indices of users in $\mathcal{K}_{i}$ in ascending order of $p_{k}$ in subsequent sections.

\section{Special Case: Clustered AI Models with Backbone Sharing}
Although we have derived the necessary conditions for the optimal user scheduling within batches of a single model, $\mathcal{P}2$ remains challenging to solve. Considering parameter sharing across models, user scheduling across batches determines the model loaded in each batch and further affects the model loading time for the next batch. Thus, $t_{n}^{\text{load}}$ in $\mathcal{P}2$ varies across batches and depends on model loading order, rendering conventional approaches for combinatorial optimization problems impractical. Specifically, with this dependency, standard DP algorithms fall short due to the need for exhaustive searches over all feasible model loading orders. Additionally, although branch-and-bound algorithms can produce the optimal solution, their worst-case time complexity is exponential with the size of the model library.
In this section, to develop a polynomial-time algorithm for real-time solution finding, we focus on an important special case called the ``backbone-sharing'' case, where the optimal user scheduling strategy can be obtained very efficiently.


We first introduce the special case. Recalling that a backbone is called a shared backbone if it appears in at least two models, we formally define the special case in Definition \ref{definition_1} and provide an illustrative example in Fig. \ref{fig:bls}.

\begin{definition}\label{definition_1}
\textbf{Clustered backbone sharing}: Consider a set of AI models that can be partitioned into disjoint clusters $\mathcal{M}=\left\{1,2\dots, M\right\}$, with $\mathcal{I}_{m}$ denoting the set of models in cluster $m$. Clustered backbone sharing occurs when each model $i\in \mathcal{I}_{m}$ consists of two parts: (1) a cluster-specific shared backbone $\mathcal{W}_{m}$ or a subset of its bottom layers, and (2) task-specific layers, which are unique to model $i$ and not shared with others. 
\end{definition}
\begin{figure}[t]
\centering
\includegraphics[width=0.4\textwidth]{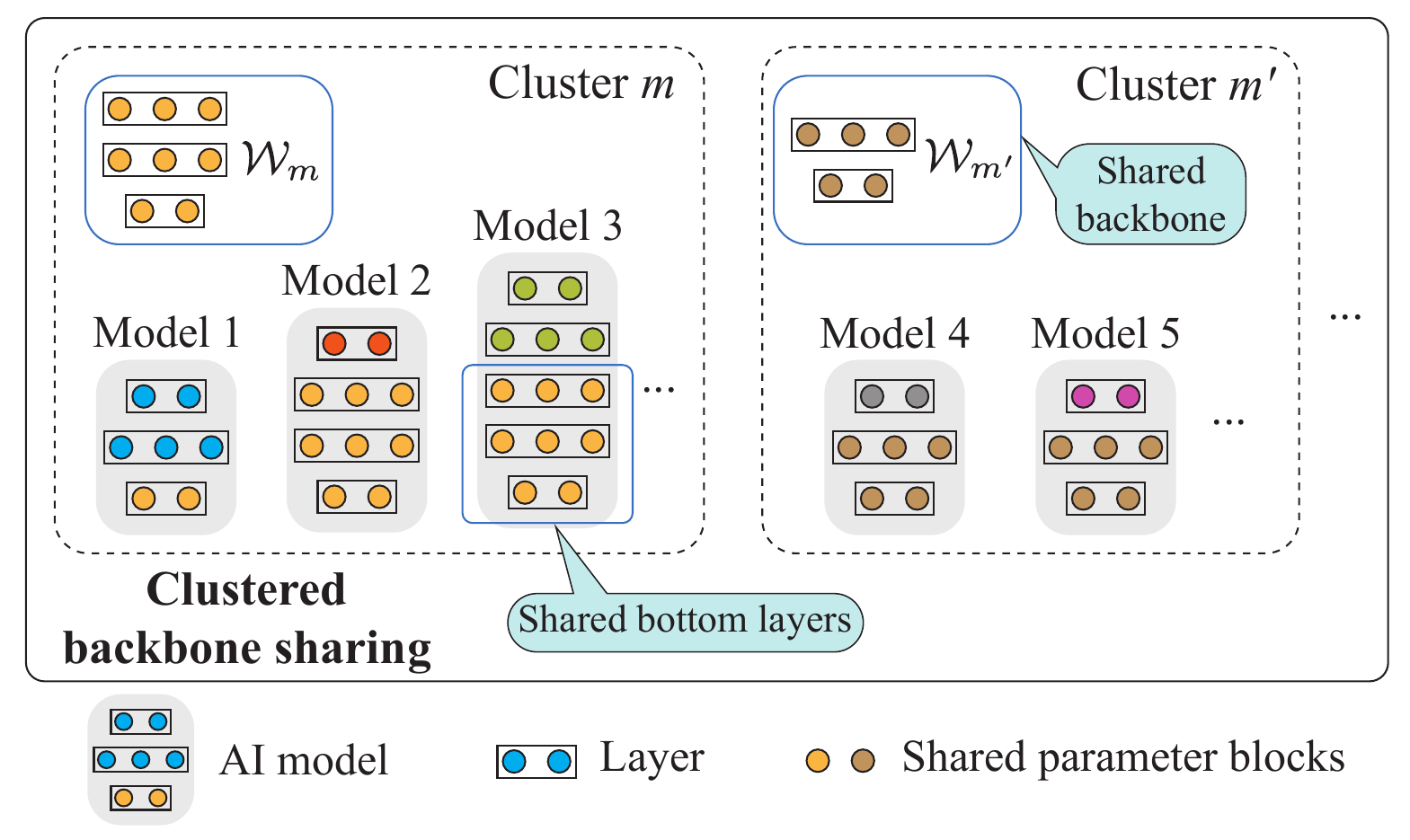}
\vspace{-3pt}
\caption{An illustrative example of backbone sharing within two model clusters, $m$ and $m'$, in the special case. Within each model cluster, models share a backbone or a subset of its bottom layers. In cluster $m$, the first layer of model 1 and the first three layers of models 2 and 3 come from the bottom layers of backbone $\mathcal{W}_{m}$. In cluster $m'$, models 4 and 5 share the entire backbone $\mathcal{W}_{m'}$.}
\vspace{-10pt}
\label{fig:bls}
\end{figure}
Note that different models in a cluster may have different numbers of shared layers, and backbone sharing across clusters is not considered, as shown in Fig. \ref{fig:bls}. 
For brevity, we use the \textit{backbone sharing (BS)} case to refer to the special case.
The BS case commonly exists in practice, as many AI models are fine-tuned from pre-trained models using bottom-layer freezing techniques \cite{10.5555/2969033.2969197,howard2018universal} or PEFT techniques~\cite{ding2023parameter}. In the first approach, the bottom layers of a pre-trained AI model are frozen while only the top layers are updated towards downstream tasks. In the second approach (e.g., adapter tuning~\cite{he-etal-2021-effectiveness}, selective PEFT~\cite{fu2023effectiveness}, and LoRA~\cite{hu2021lora}), one can freeze the entire pre-trained model and update only a small set of parameters during fine-tuning. Moreover, in PEFT settings, the entire shared pre-trained model and all fine-tuned parameters can be mathematically treated as a single shared bottom layer and a single
task-specific layer, respectively, according to Definition~\ref{definition_1}. The rationale of these two approaches comes from the fact that the general knowledge embedded in pre-trained models can be reused in downstream tasks\footnote{For example, the bottom layers in CNNs typically extract low-level features (e.g., edges, corners, and basic shapes), which are transferable and reusable for various downstream tasks.}.


Beyond Theorem \ref{proposition_2}, we provide an additional necessary condition for the optimal user scheduling solution $\tilde{{\bf{X}}}^{*}$ to $\mathcal{P}2$ in the BS case.

\begin{theorem}\label{proposition_3}
    In $\tilde{{\bf{X}}}^{*}$, users requesting models from the same cluster are scheduled in the same or consecutive batches.
\end{theorem}
\begin{proof}
    The proof is provided in Appendix \myref{proof:proposition_3}.
\end{proof}



With Theorem \ref{proposition_3}, we have the following corollary.
\begin{corollary}\label{proposition_3_3}
    \textbf{The optimal model loading strategy:} The edge server schedules the model loading order across clusters in ascending order of cluster index $m$. For each cluster $m$, it schedules the model loading order for models in $\mathcal{I}_{m}$ in ascending order of $l_{i}$, where $i\in\mathcal{I}_{m}$. 
\end{corollary}
\begin{proof}  
    The proof is provided in Appendix \myref{proof:proposition_3_3}.
\end{proof}
\begin{remark}
    Corollary \ref{proposition_3_3} guarantees parameter blocks are never redundantly loaded in the BS case. It enables us to determine the model loading order before scheduling users for each model, thus significantly simplifying the problem.
\end{remark}
Thereafter, we reorder the models in $\mathcal{I}_{m}$ by ascending order of $l_{i}$ for ease of presentation. In the next subsection, the user scheduling strategy will be found with the above strategy. 


\subsection{The Solution Approach}
With the optimal model loading order from Corollary \ref{proposition_3_3}, this subsection will develop a DP-based algorithm to compute the maximum task throughput and design a recursive algorithm to obtain the optimal user-scheduling decision $\tilde{{\bf{X}}}^{*}$ in the BS case. 
The details are as follows.
\subsubsection{Number of served users of a model} Assuming that model $i$ is loaded into the GPU memory after model $\hat{i}$, based on Theorem \ref{proposition_2} and Corollary \ref{proposition_3_3}, we first determine the maximum number of users that can be served by model $i$ in cluster $m$ within $\hat{\tau}$ time slots, which is denoted as
\begin{equation}\label{eq_model_user}
\begin{aligned}
    &q_{m}\left(\hat{i},i,\hat{\tau}\right) = \\
    &\begin{cases}
    \begin{aligned}
    &\max \left\{\hat{k}\left|
            \begin{aligned}
                \Phi_{m}\left(\hat{i},i,\hat{k}\right)
                \le \hat{\tau}\Delta\tau
            \end{aligned}
            \right.
            \right\},\text{if}\ \hat{\tau}\in\left[1,T\right],i\in\mathcal{I}_{m},\\
            &0,\ \text{if}\ i = 0\text{ or } \hat{\tau}=0,
            \end{aligned}
    \end{cases}
\end{aligned}
\end{equation}
where $\Phi_{m}\left(\hat{i},i,\hat{k}\right)$ is the total latency required to perform inference for $\hat{k}$ users with model $i$ loaded after model $\hat{i}$ within model cluster $m$, and is given by
\begin{equation}\label{eq_Phi}
\begin{aligned}
    &\Phi_{m}\left(\hat{i},i,\hat{k}\right)= 
    \sum\limits_{n=1}^{N'_{i}\left(\hat{k}\right)}p_{z\left(n,i,\hat{k}\right)}
    +V\left(\sum\limits_{l=l_{\hat{i}}+1}^{L_{i}}S'_{i,l}\right)\\
    &+\sum\limits_{n=1}^{N'_{i}\left(\hat{k}\right)}e\left(n,i,\hat{k}\right), \hat{k}\in\left[0,\left|\mathcal{K}_{i}\right|\right], \hat{i}\in\left[0,i-1\right],  \hat{i}\in\left\{0\right\}\cup\mathcal{I}_{m}.
\end{aligned}
\end{equation}
In \eqref{eq_Phi}, the first term $\sum\limits_{n=1}^{N'_{i}\left(\hat{k}\right)}p_{z\left(n,i,\hat{k}\right)}$ is the data uploading time for the first $\hat{k}$ users in $\mathcal{K}_{i}$ across multiple batches. Here, $z\left(n,i,\hat{k}\right) = \min\left\{nb_{i,\max},\hat{k}\right\}$ is the index of the user with the highest uploading time in batch $n$, where $b_{i,\max}$ is the maximum allowed batch size for model $i$, determined by the GPU memory capacity and the peak memory required by model $i$.  
$N'_{i}\left(\hat{k}\right)=\lceil\frac{\hat{k}}{b_{i,\max}}\rceil$ is the total number of batches required to serve $\hat{k}$ users with model $i$.
Furthermore, the second term $V\left(\sum\limits_{l=l_{\hat{i}}+1}^{L_{i}}S'_{i,l}\right)=V_{1}\left(\sum\limits_{l=l_{\hat{i}}+1}^{L_{i}}S'_{i,l}\right)+V_{2}\left(\sum\limits_{l=l_{\hat{i}}+1}^{L_{i}}S'_{i,l}\right)$ is the total model loading time for the top $L_{i}-l_{\hat{i}}$ layers of model $i$, which are not shared with model $\hat{i}$. It includes loading time from the disk to system memory and from system memory to GPU memory. Here, $S'_{i,l}$ is the size of the $l$-th layer of model $i$ and $l_{0}=0$. 
Moreover, the last term $\sum\limits_{n=1}^{N'_{i}\left(\hat{k}\right)}e\left(n,i,\hat{k}\right)$ is the total inference computing time for $\hat{k}$ users served in consecutive batches with model $i$. Here, based on Theorem \ref{proposition_2}, $e\left(n,i,\hat{k}\right)$ is given as
\begin{equation}
\begin{aligned}
    e\left(n,i,\hat{k}\right) = \mu_{i}\min\left\{b_{i,\max},\hat{k}-\left(n-1\right)b_{i,\max}\right\}+\beta_{i},
\end{aligned}
\end{equation}
where $\min\left\{b_{i,\max},\hat{k}-\left(n-1\right)b_{i,\max}\right\}$ is the batch size of batch $n$, determined by $b_{i,\max}$ and the remaining number of users after scheduling the first $n-1$ batches.

\subsubsection{Number of served users in a model cluster}We use $g_{m}\left(i,\tau_{2}\right)$ to represent the maximum number of users that can be served with the first $i$ models in $ \mathcal{I}_{m}$ within the first $\tau_{2}$ time slots of a time period consisting of $\tilde{\tau}$ time slots in total. Based on Theorem \ref{proposition_2} and Corollary \ref{proposition_3_3}, $g_{m}\left(i,\tau_{2}\right)$ is expressed as
\begin{equation}\label{eq_incluster_user}
    g_{m}\left(i,\tau_{2}\right) = 
    \begin{cases}
                    \begin{aligned}
                        &\Gamma_{m}\left(i,\tau_{2}\right), \ \text{if}\ \tau_{2}\in\left[1,\tilde{\tau}\right] ,\  i\in\mathcal{I}_{m}\\
                        &0,\ \text{if}\ i = 0\text{ or }\tau_{2} = 0,
                    \end{aligned}
		\end{cases}
\end{equation}
where 
$\Gamma_{m}\left(i,\tau_{2}\right)$ is provided in \eqref{eq_gamma} at the bottom of next page. 

\begin{figure*}[!b]
\begin{equation}\label{eq_gamma}
    \Gamma_{m}\left(i,\tau_{2}\right)=\max \ \left\{
                        \begin{aligned}
                            &\gamma_{m}\left(i,\tau_{2}\right)=\mathop{\max}\limits_{
                            \hat{i}\in\left(0,i-1\right],\hat{i}\in\mathcal{I}_{m},\hat{\tau}\in\left[1,\tau_{2}\right)
                            }\ \left\{
                            \begin{aligned}
                            &g_{m}\left(\hat{i},\tau_{2}-\hat{\tau}\right)\\
                            &+q_{m}\left(\hat{i},i,\hat{\tau}\right)
                            \end{aligned}
                            \ \middle| \
                            \begin{aligned}
                                &g_{m}\left(\hat{i},\tau_{2}-\hat{\tau}\right)>0,\\ &g_{m}\left(\hat{i},\tau_{2}-\hat{\tau}\right) \ne g_{m}\left(\hat{i}-1,\tau_{2}-\hat{\tau}\right)
                            \end{aligned}
                            \right\},\\
                            &q_{m}\left(\hat{i}=0,i,\hat{\tau}=\tau_{2}\right),\ g_{m}\left(i-1,\tau_{2}\right)
                        \end{aligned}
                        \right\}
\end{equation}
\end{figure*}

In $\Gamma_{m}\left(i,\tau_{2}\right)$, three terms correspond to three cases considered for model $i$ in cluster $m$. 
For the first term $\gamma_{m}\left(i,\tau_{2}\right)$, we consider that model $\hat{i}$ is the preceding model of model $i$, and model $i$ is allocated to $\hat{\tau}$ time slots to serve $q_{m}\left(\hat{i},i,\hat{\tau}\right)$ users, where $g_{m}\left(\hat{i},\tau_{2}-\hat{\tau}\right)$ is the number of users served by the first $\hat{i}$ models in cluster $m$ within $\tau_{2}-\hat{\tau}$ time slots. 
For the second term $q_{m}\left(\hat{i}=0,i,\hat{\tau}=\tau_{2}\right)$, we consider that model $i$ does not have a preceding model (i.e., $\hat{i}=0$), and all $\tau_{2}$ time slots are allocated to model $i$ (i.e., $\hat{\tau}=\tau_{2}$).
For the third term $g_{m}\left(\hat{i}=i-1,\tau_{2}\right)$, we consider that there is no need to load model $i$, leading to $g_{m}\left(i,\tau_{2}\right)=g_{m}\left(i-1,\tau_{2}\right)$, as the number of served users remains the same.


\subsubsection{Number of served users across all model clusters}We use $f\left(m,\tau_{1}\right)$ to represent the maximum number of users that can be served with the first $m$ model clusters within the first $\tau_{1}$ time slots. Based on Theorem \ref{proposition_3} and Corollary \ref{proposition_3_3}, $f\left(m,\tau_{1}\right)$ can be given as 
\begin{equation}\label{eq_cluster_user}
\begin{aligned} 
    &f\left(m,\tau_{1}\right) = \\
    &\begin{cases}
    \begin{aligned}
    &\mathop{\max}\limits_{\tilde{\tau}\in \left[0,\tau_{1}\right]}\ \left\{
    \begin{aligned}
        &f\left(m-1,\tau_{1}-\tilde{\tau}\right) \\
        &+ g_{m}\left(I_{m},\tilde{\tau}\right)
    \end{aligned}\right\}, \text{if}\ \tau_{1}\in\left[1,T\right], m\in\mathcal{M},\\
    &0,\ \text{if}\ m=0\text{ or } \tau_{1} = 0,
    \end{aligned}
    \end{cases}
\end{aligned}
\end{equation}  
where $g_{m}\left(I_{m},\tilde{\tau}\right)$, as recalled, represents the maximum number of users that can be served by the first $I_{m}$ models in cluster $m$ in multiple consecutive batches within $\tilde{\tau}$ time slots.

After traversing all clusters and time slots in $f\left(m,\tau_{1}\right)$, the maximum number of served users can be determined by
\begin{equation}
    U\left(\tilde{{\bf{X}}}^{*}\right) = f\left(M,T\right),
\end{equation}
where $\tilde{{\bf{X}}}^{*}=\bigcup\limits_{n=1}^{N}\tilde{{\bf{X}}}^{*}_{n}$ is the determined user scheduling in the BS case, with $\tilde{x}^{*}_{n,k}\in\tilde{{\bf{X}}}^{*}_{n}$, which will be elaborated next. 

\subsubsection{Determination of $\tilde{{\bf{X}}}^{*}$} The user scheduling decision corresponding to $U\left(\tilde{{\bf{X}}}^{*}\right)$ can be obtained by Algorithm \ref{algorithm_recursive}. In the outer loop (Line \ref{line:rec_while_1_s} to \ref{line:rec_while_1_e}) over cluster $m$, we focus on identifying the maximum number of users served by models in cluster $m$, represented by $g_{m}\left(I_{m},\tilde{\tau}^{*}\right)$. We determine the corresponding assigned time slots $\tilde{\tau}^{*}$ for cluster $m$ in Line \ref{line:rec_6} within the remaining $\tau_{1}$ time slots by traversing the clusters in reverse order. In the inner loop (Line \ref{line:rec_while_2_s} to \ref{line:rec_while_2_e}) over model $i$ in cluster $m$, we first identify the optimal preceding model $\hat{i}^{*}$ for model $i$ in cluster $m$ along with the corresponding assigned model slots $\hat{\tau}^{*}$ for model $i$ in Line \ref{line:rec_12} in the reverse order of models in $\mathcal{I}_{m}$. Then, we compute the maximum number of users, $\hat{k}^{*}$, that can be served using model $i$ within $\hat{\tau}^{*}$ time slots, assuming model $i$ is loaded after $\hat{i}^{*}$ in Line \ref{line:rec_13}. Next, we calculate the number of batches for serving $\hat{k}^{*}$ users with model $i$ in Line \ref{line:rec_n'}. Finally, we schedule the first $\hat{k}^{*}$ users in $K_{i}$ across multiple batches and insert these batches at the start of $\tilde{{\bf{X}}}^{*}_{N}$ from Line \ref{line:rec_n_s} to \ref{line:rec_n_e}. Note that there may be more than one user scheduling decision that can achieve the maximum throughput, and Algorithm \ref{algorithm_recursive} produces one of such solutions by starting from the maximum task throughput and identifying the feasible maximum number of served users in each loop.

\begin{algorithm}[!t]
	\caption{Recursive Algorithm for Determining User Scheduling} 
	\label{algorithm_recursive}
	\LinesNumbered
	\KwIn{$q_{m}\left(\hat{i},i,\hat{\tau}\right)$, $g_{m}\left(i,\tau_{2}\right)$, and $f\left(m,\tau_{1}\right)$ calculated in Algorithm \ref{algorithm_DP}. $\mathcal{I}_{m}$ and $\mathcal{K}$.}
	\KwOut{$\tilde{{\bf{X}}}^{*}$.} 
	{\bf Initialize:} $m = M$, $\tau_{1}=T$, $N=0$, $\tilde{{\bf{X}}}^{*}_{N}=\emptyset$.\\
        \While{$m>0$}
        {\label{line:rec_while_1_s}
            \If{$\tau_{1}\le 0$}
            {
                \textbf{Break}.\\
            }
            $\tilde{\tau}^{*}=\mathop{\arg\max}\limits_{\tilde{\tau} \in \left[0,\tau_{1}\right]}\left\{f\left(m-1,\tau_{1}-\tilde{\tau}\right) + g_{m}\left(I_{m},\tilde{\tau}\right)\right\}$.\label{line:rec_6}\\
            $i=I_{m}$ and $\tau_{2}=\tilde{\tau}^{*}$.\\
            \While{$i>0$ }
            {\label{line:rec_while_2_s}
                \If{$\tau_{2}\le0$}
                {
                    \textbf{Break}.\\
                }
                $\left\{\hat{i}^{*},\hat{\tau}^{*}\right\}=\mathop{\arg\max}\limits_{\hat{i},\hat{\tau}}\ 
                \Gamma_{m}\left(i,\tau_{2}\right)$.\label{line:rec_12}\\
                $\hat{k}^{*} = q_{m}\left(\hat{i}^{*},i,\hat{\tau}^{*}\right)$. \label{line:rec_13}\\
                $N'_{i}\left(\hat{k}^{*}\right)=\lceil \frac{\hat{k}^{*}}{b_{i,\max}}\rceil$.\label{line:rec_n'}\\
                \If{$N\ne 0$}
                {\label{line:rec_n_s}
                    $n=N$.\\
                    \For{$1\le n \le N$}
                {
                    $\tilde{{\bf{X}}}^{*}_{n+N'_{i}\left(\hat{k}^{*}\right)} = \tilde{{\bf{X}}}^{*}_{n}$. 
                    $n = n - 1$.\\
                }
                }\label{line:rec_n_m1}
                \For{$1\le n \le N'_{i}\left(\hat{k}^{*}\right)$}
                {\label{line:rec_n_m2}
                    $\tilde{{\bf{X}}}^{*}_{n} =\left\{\tilde{x}^{*}_{n,k}\right\}$, where $k\in\mathcal{K}_{i}$ and $k\in \left\{
                        \begin{aligned}
                            &\left(n-1\right)b_{i,\max}+1,\dots,\left(n-1\right)b_{i,\max}\\
                            &+\min\left\{b_{i,\max},\hat{k}^{*}- \left(n-1\right)b_{i,\max}\right\}
                        \end{aligned}\right\}$.\\
                    $n = n + 1$.\\
                }\label{line:rec_n_e}
                $N=N+ N'_{i}\left(\hat{k}^{*}\right)$, $i = \hat{i}^{*}$, $\tau_{2} = \tau_{2}-\hat{\tau}^{*}$.\\
            }\label{line:rec_while_2_e}
            $m = m-1$, $\tau_{1} = \tau_{1}-\tilde{\tau}^{*}$.\\
        }\label{line:rec_while_1_e}
        $\tilde{{\bf{X}}}^{*}=\bigcup\limits_{n=1}^{N}\tilde{{\bf{X}}}^{*}_{n}$.\\
\end{algorithm}


\subsubsection{Algorithm outline and analysis} To solve $\mathcal{P}1$ in the BS case, we summarize the entire process in Algorithm \ref{algorithm_DP}, which calculates $\tilde{{\bf{X}}}^{*}$ from Algorithm \ref{algorithm_recursive} and derives the optimal bandwidth allocation $\tilde{y}^{*}_{n,k}\in\tilde{{\bf{Y}}}^{*}$ according to \eqref{optimal_bandwidth}. Furthermore, we establish the following theorem for Algorithm \ref{algorithm_DP}. 
\begin{algorithm}[!t]
	\caption{DP-based Algorithm} 
	\label{algorithm_DP}
	\LinesNumbered
	\KwIn{$\mathcal{K}$, $\mathcal{I}$, $\mathcal{M}$, $T$.}
	\KwOut{$\tilde{{\bf{X}}}^{*}$, $\tilde{{\bf{Y}}}^{*}$, $U\left(\tilde{{\bf{X}}}^{*}\right)$.} 
	{\bf Initialize:} 
    Set $q_{m}\left(\hat{i},i,\hat{\tau}\right)$, $g_{m}\left(i,\tau_{2}\right)$, $f\left(m,\tau_{1}\right)$ to 0. \\\label{line:dp_s}
        \For{$1\le m \le M$}
        {\label{line:DP_m1_s}
            \For{$0 \le i\le I_{m}$}
            {
                \For{$0\le\hat{\tau}\le T$}
                {
                    \For{$0\le \hat{i} \le i-1$}
                    {
                        Calculate $q_{m}\left(\hat{i},i,\hat{\tau}\right)$ from \eqref{eq_model_user}.\\
                    }
                }
            }
        }\label{line:DP_m1_e}
        \For{$1\le m \le M$}
        {\label{line:DP_m2_s}
            \For{$0\le i \le I_{m}$}
            {
                \For{$0\le\tau_{2}\le T$}
                {
                    {
                        {
                        Calculate $g_{m}\left(i,\tau_{2}\right)$ from \eqref{eq_incluster_user}.
                        }
                    }
                }
            }
        }\label{line:DP_m2_e}
	\For{$0\le m \le M$}
	{\label{line:DP_m3_s}
            \For{$0\le\tau_{1}\le T$}
            {
                {
                    Calculate $f\left(m,\tau_{1}\right)$ from \eqref{eq_cluster_user}.\\
                }
            }
        }\label{line:DP_m3_e}
        $U\left(\tilde{{\bf{X}}}^{*}\right) = f\left(M,T\right)$.\\
        With $q_{m}\left(\hat{i},i,\hat{\tau}\right)$, $g_{m}\left(i,\tau_{2}\right)$, and $f\left(m,\tau_{1}\right)$, calculate $\tilde{{\bf{X}}}^{*}$ from Algorithm \ref{algorithm_recursive}.\label{line:dp_x_e}\\
        Calculate $\tilde{{\bf{Y}}}^{*}$ using $\tilde{y}^{*}_{n,k}=\frac{\tilde{x}^{*}_{n,k}D_{k}}{B\bar{R}_{k}\sum\limits_{k'\in\mathcal{K}}\frac{\tilde{x}^{*}_{n,k'}D_{k'}}{B\bar{R}_{k'}}}$.\label{line:DP_y}\\
\end{algorithm}
\begin{theorem}\label{theorem_spec}
    The proposed Algorithm \ref{algorithm_DP} achieves the optimal solution to $\mathcal{P}1$ in the BS case with a polynomial-time computational complexity $O\left(MI^{2}K\right)$.
\end{theorem}
\begin{proof}
    We begin by proving that Algorithm \ref{algorithm_DP} achieves the optimal solution in the BS case. First, based on Theorem \ref{proposition_2}, \eqref{eq_model_user} computes the maximum number of served users requesting model $i$ within the scheduled $\hat{\tau}$ time slots, and \eqref{eq_incluster_user} computes the maximum number of served users requesting models in $\mathcal{I}_{m}$. Second, by Corollary \ref{proposition_3_3}, \eqref{eq_cluster_user} computes the maximum number of served users requesting models in $\mathcal{I}$. Third, the proposed algorithm adopts a DP scheme, which guarantees optimality by traversing all models within each model cluster and across all model clusters using \eqref{eq_incluster_user} and \eqref{eq_cluster_user}, respectively. Therefore, $\tilde{{\bf{X}}}^{*}$, obtained in Line \ref{line:dp_x_e} of Algorithm \ref{algorithm_DP}, is the optimal user scheduling for $\mathcal{P}2$. Furthermore, based on Proposition \ref{proposition_1} and Theorem \ref{theorem_1}, $\tilde{{\bf{X}}}^{*}$ and $\tilde{{\bf{Y}}}^{*}$ are the optimal user scheduling and spectrum bandwidth allocation for $\mathcal{P}1$, respectively.
    
    Next, we show the time complexity of Algorithm \ref{algorithm_DP} is $O\left(MI^{2}K\right)$. First, the time complexity from Line \ref{line:DP_m1_s} to \ref{line:DP_m3_e} in Algorithm \ref{algorithm_DP} is $O\left(TMI^{2}K + T^{2}MI^{2}+ T^{2}M \right) = O\left(MI^{2}K\right)$, where $O\left(TMI^{2}K\right)$ is the time complexity of Line \ref{line:DP_m1_s} to \ref{line:DP_m1_e}, $O\left(T^{2}MI^{2}\right)$ accounts for Line \ref{line:DP_m2_s} to \ref{line:DP_m2_e}, and $O\left(T^{2}M\right)$ corresponds to Line \ref{line:DP_m3_s} to \ref{line:DP_m3_e} in Algorithm \ref{algorithm_DP}. 
    Second, the time complexity of Algorithm \ref{algorithm_recursive} is $O\left(MT+MI\left(IT+1+K\right)\right) = O\left(TMI^{2}+MIK\right)$, where $O\left(MT\right)$ corresponds to Line \ref{line:rec_6} in Algorithm \ref{algorithm_recursive},  $O\left(MI\left(IT+1+K\right)\right)$ accounts for Line \ref{line:rec_12}, Line \ref{line:rec_13}, and Line \ref{line:rec_n_s} to \ref{line:rec_n_e} in Algorithm \ref{algorithm_recursive}. Note that the time complexity of Line \ref{line:rec_13} in Algorithm \ref{algorithm_recursive} is $O\left(1\right)$, as the value of $q_{m}\left(\hat{i}^{*},i,\hat{\tau}^{*}\right)$ has been computed in Algorithm \ref{algorithm_DP} before executing Algorithm \ref{algorithm_recursive}. Third, the time complexity of computing $\tilde{\bf{Y}}$ in Line \ref{line:DP_y} of Algorithm \ref{algorithm_DP} is $O\left(K\right)$. 
    Therefore, the total time complexity of Algorithm \ref{algorithm_DP} is $O\left(MI^{2}K\right)$, which completes the proof.
\end{proof}

\begin{remark}\label{remark_independent}
    \textbf{Optimality in independent model loading}: Theorem \ref{theorem_spec} remains valid for multi-user edge inference systems when parameter sharing is not exploited. In such scenarios, all parameter blocks of an AI model are fully loaded during the model loading step of an inference batch, and Theorem \ref{proposition_3} and Corollary \ref{proposition_3_3} continue to ensure optimality. Therefore, while Algorithm \ref{algorithm_DP} is designed for the BS case, it also guarantees an optimal solution for independent model loading with the same computational time complexity.
\end{remark}
\section{General Case}
This section addresses the general case of $\mathcal{P}1$. In the general case, we do not impose any assumptions on the shared parameter blocks across models; these blocks can appear at arbitrary positions within a model rather than forming consecutive bottom layers of the backbone.
Although $\mathcal{P}1$ can still be decoupled into solving $\mathcal{P}2$ followed by calculating the optimal spectrum bandwidth allocation for any given user scheduling solution using \eqref{optimal_bandwidth} in the general case, Theorem \ref{proposition_3} and Corollary \ref{proposition_3_3}, which hold in the BS case, do not apply to $\mathcal{P}2$ in the general case. As $\mathcal{P}2$ is NP-hard, finding the optimal model loading order across batches is computationally intractable in polynomial time. Consequently, finding the optimal user scheduling in the general case is not computationally feasible. Therefore, we will develop a heuristic algorithm to obtain a sub-optimal user scheduling solution to $\mathcal{P}2$ in the general case and compute the corresponding optimal spectrum bandwidth allocation.

Our heuristic employs a greedy procedure. Specifically, we aim to select the next model with the maximum number of users served per unit time slot in each iteration. Given that model $i'$ is loaded before model $i$, the number of users served by model $i$ per unit time slot within $\tau$ time slots is given by 
\begin{equation}
\begin{aligned}
    &v\left(i',i,\tau\right) =\\
    &\begin{cases}
    \begin{aligned}
    &\max \left\{\frac{k}{\tau} \ \middle| \
            \begin{aligned}
                \phi\left(i',i,\tau\right)
                \le \tau\Delta\tau
            \end{aligned}
            \right\},  \text{if}\ \tau\in\left[1,T\right], i\in\mathcal{I},\\
    &0,\ \text{if} \ i=0\text{ or }\tau=0.    
\end{aligned}
    \end{cases}
\end{aligned}
\end{equation}

Analogous to \eqref{eq_Phi}, $\phi\left(i',i,\tau\right)$ is the total inference latency for $k$ users with model $i$ loaded after model $i'$ in the general case, which is
\begin{equation}
\begin{aligned}&\phi\left(i',i,\tau\right)=\sum\limits_{n=1}^{N'_{i}\left(k\right)}p_{z\left(n,i,k\right)}+V\left(\sum\limits_{j\in\mathcal{J}_{i}\setminus\mathcal{J}_{i'}}S_{j}\right)\\
    &+\sum\limits_{n=1}^{N'_{i}\left(k\right)}e\left(n,i,k\right), k\in\left[0,\left|\mathcal{K}_{i}\right|\right], i'\in\mathcal{I}\setminus\left\{i\right\},
\end{aligned}
\end{equation}
where $\sum\limits_{j\in\mathcal{J}_{i}\setminus\mathcal{J}_{i'}}S_{j}$ is the data size of parameter blocks of model $i$ but not of model $i'$, which need to be loaded from the disk to system memory and then to GPU memory in the general case. $\mathcal{J}_{i}$ is the set of parameter blocks of model $i$, with $\mathcal{J}_{0}=\emptyset$. 

The proposed algorithm is outlined in Algorithm \ref{algorithm_greedy}. In the first loop over $\tau'$, we find the next model $i^{*}$ and the corresponding assigned time slots $\tau^{*}$ by identifying the maximum number of users served per unit time slot in Line \ref{line:greedy_i_tau}. Next, we calculate the number of users served by model $i^{*}$ within $\tau^{*}$ time slots, denoted by $k^{*}$, in Line \ref{line:greedy_k}. Then, in Line \ref{line:greedy_n}, we calculate the number of batches required to serve $k^{*}$ users with model $i^{*}$. Finally, from Line \ref{line:greedy_n_s} to Line \ref{line:greedy_n_e}, we schedule the first $k^{*}$ users in $\mathcal{K}_{i}$ into multiple batches according to Theorem \ref{proposition_2_1}, appending these batches to the end of ${\bf{X}}_{N}$. Algorithm \ref{algorithm_greedy} exhibits a polynomial-time complexity as follows.
\begin{theorem}
    The proposed Algorithm \ref{algorithm_greedy} solves $\mathcal{P}1$ with a time complexity of $O\left(K\right)$.
\end{theorem}
\begin{proof}
    The time complexity of Line \ref{line:greedy_i_tau} in Algorithm \ref{algorithm_greedy} is $O\left(I\right)$, and that of Line \ref{line:greedy_n_s} to \ref{line:greedy_n_e} is $O\left(K\right)$. Additionally, the time complexity of Line \ref{line_greedy_y} in Algorithm \ref{algorithm_greedy} is $O\left(K\right)$. Therefore, the total time complexity of Algorithm \ref{algorithm_greedy} is $O\left(\max\left\{I,K\right\}\right) = O\left(K\right)$. Here, the equality holds since one AI model in the model library is requested by at least one user, which completes the proof.
\end{proof}

\begin{algorithm}[!t]
	\caption{Greedy Algorithm} 
	\label{algorithm_greedy}
	\LinesNumbered
	\KwIn{$\mathcal{I}$, $\mathcal{K}$, $T$.}
	\KwOut{${\bf{X}}$, ${\bf{Y}}$, $U\left({\bf{X}}\right)$.} 
	{\bf Initialize:} $\mathcal{I}'= \mathcal{I}$, $\tau'=0$, $N=0$, $\tilde{{\bf{X}}}^{*}_{N}=\emptyset$, $i'=0$.\\
        
        \While{$\tau' < T$}
        {\label{line:greedy_t_s}
            \If{$\mathcal{I}' = \emptyset$}
            {
                \textbf{Break}.\\
            }
            $\left\{i^{*},\tau^{*}\right\}=\mathop{\arg\max}\limits_{i\in\mathcal{I}',\tau\in\left(0,T - \tau'\right]}\left\{v\left(i',i,\tau\right)\right\}$ \label{line:greedy_i_tau} \\ 
            \If{$i^{*}$ does not exist}
            {
                \textbf{Break}.
            }
            $k^{*} = v\left(i',i^{*},\tau^{*}\right)\tau^{*}$.\label{line:greedy_k}\\
            
            $N'_{i}\left(k^{*}\right)=\lceil \frac{k^{*}}{b_{i^{*},\max}}\rceil$.\label{line:greedy_n}\\
            \For{$1\le n \le N'_{i}\left(k^{*}\right)$}
            {\label{line:greedy_n_s}
                ${\bf{X}}_{n+N} = \left\{x_{n,k}\right\}$, where $k\in\mathcal{K}_{i}$ and 
                $k\in\left\{
                \begin{aligned}
                    &\left(n-1\right)b_{i^{*},\max}+1,\dots,\left(n-1\right)b_{i^{*},\max}\\
                    &+\min\left\{b_{i^{*},\max},k^{*}- \left(n-1\right)b_{i^{*},\max}\right\}
                \end{aligned}\right\}$.\\
                $n = n + 1$.\\
            }\label{line:greedy_n_e}
            $i'=i^{*}$, $\mathcal{I}'=\mathcal{I}'\setminus\left\{i^{*}\right\}$, $\tau'=\tau'+\tau^{*}$, $N=N+N'_{i}\left(k^{*}\right)$.\\
        }\label{line:greedy_t_e}
        ${\bf{X}}=\bigcup\limits_{n=1}^{N}{\bf{X}}_{n}$ and calculate $U\left({\bf{X}}\right)$.\\
        Calculate ${\bf{Y}}$ using \eqref{optimal_bandwidth}.\\\label{line_greedy_y}
\end{algorithm}

\section{Numerical Results}
\subsection{Simulation Setup}
In our simulation settings, the coverage radius of the edge server is 250 m, with $K$ users evenly distributed in the coverage area of the edge server, where $K$ is chosen from \{60, 70, 80, 90, 100\}. $P_{k}$, $N_{0}$, and $\alpha$ of $\bar{R}_{k}$ in \eqref{eq:latency_uploading} are set to $\text{5}\times\text{10}^{\text{-9}}$ Watt/Hz, -174 dBm/Hz, and 4, respectively. The total bandwidth $B$ ranges from \{10, 50, 100, 200, 300, 400\} MHz. The latency requirement $\bar{T}$ varies between 500 and 900 ms, with each time slot lasting 10 ms.

An AI model library is constructed using three model structures: ResNet-18, ResNet-34, and ResNet-50. We first pre-train the three models on CIFAR-100 and then fine-tune the pre-trained models on CIFAR-10. In the BS case, the bottom-layer freezing technique is applied to fine-tune 50 AI models. The fine-tuned models share consecutive bottom layers with the corresponding pre-trained model of the same model structure. As a result, models fine-tuned from the same pre-trained model are grouped into a single model cluster and share bottom layers with the other models within the same cluster. In the general case, 25 AI models fine-tuned from the pre-trained model share parameter blocks at arbitrary positions. Moreover, in both the BLS and general cases, models share an average of $\theta$ = \{75\%, 80\%, 85\%, 90\%, 95\%\} of layers with the corresponding pre-trained models. For different values of $\theta$, compared with full-parameter fine-tuning, the inference accuracy changes, on average, by \{0.03\%, -0.03\%, 0.09\%, 0.19\%, -0.63\%\} in the BS case and by \{0.02\%, -0.42\%, -0.49\%, -0.09\%, -1.15\%\} in the general case, with positive values indicating accuracy enhancement by overcoming overfitting. The dataset for inference computing is CIFAR10, where the original data size is resized from 32$\times$32$\times$3 to 128$\times$128$\times$3 to simulate higher-resolution data uploading.

Three algorithms are compared in the simulation.
\begin{itemize}
    \item \textbf{PartialLoading Spec}. Algorithm \ref{algorithm_DP}, designed to obtain the optimal solution in the BS case.
    \item \textbf{PartialLoading Gen}. Algorithm \ref{algorithm_greedy}, developed for the general case. It is also applicable to the BS case. 
    \item \textbf{Independent Loading}. The \textit{optimal} user scheduling and spectrum bandwidth allocation policy for $\mathcal{P}1$ without considering parameter sharing. As noted in Remark \ref{remark_independent}, Algorithm \ref{algorithm_DP} also produces the optimal solution when parameter sharing is not considered. This baseline serves as the performance upper bound for traditional schemes without exploiting parameter sharing.
\end{itemize}

We consider Rayleigh fading channels in the simulations. The results are averaged from $10^{3}$ channel realizations. To evaluate the performance of the proposed approaches, we use the served user ratio as the metric, which is defined as the ratio of the number of users successfully served to the total number of users subject to the latency constraints. This metric is normalized to $[0,1]$. In addition, we use a Linux edge server with configurations consistent with Fig. \ref{fig:intro_comp}.

\subsection{Performance in the BS case}
We first evaluate the performance of the proposed approaches in the BS case. Although PartialLoading Gen is designed for the general case, it can also be applied to the BS case and is therefore included as a benchmark. 

In Fig. \ref{fig:special}, we compare the performance by varying the total bandwidth $B$, number of users $K$, latency constraint $\bar{T}$, and layer sharing ratio $\theta$. Fig. \ref{fig:spec_bw} demonstrates that increasing the value of $B$ can increase the served user ratio. This is because the uploading time in \eqref{eq:latency_uploading} decreases as more communication resources are allocated to users. Notably, when $B$ ranges from 10 to 100 MHz, the served user ratio increases significantly with growing bandwidth. However, the increasing trend slows when $B$ is larger than 100 MHz. It illustrates that the performance bottleneck of the proposed framework lies in the data uploading when communication resources are limited. In contrast, the bottleneck shifts to model loading and inference computing when communication resources are abundant. 
This insight can guide network operators in allocating communication and computing resources effectively to balance performance across different stages of the framework.  
Moreover, both PartialLoading Spec and PartialLoading Gen outperform Independent Loading. Specifically, PartialLoading Spec and PartialLoading Gen improve the served user ratio by about 24\% and 16\%, respectively, on average. Note that PartialLoading Spec obtains the optimal solution in the BS case. As a result, the small performance gap between PartialLoading Spec and PartialLoading Gen indicates that PartialLoading Gen also performs well.

\begin{figure}[t]
    \centering
	\subfigure[Served user ratio vs. $B$.]{\includegraphics[height=2.9cm, keepaspectratio]{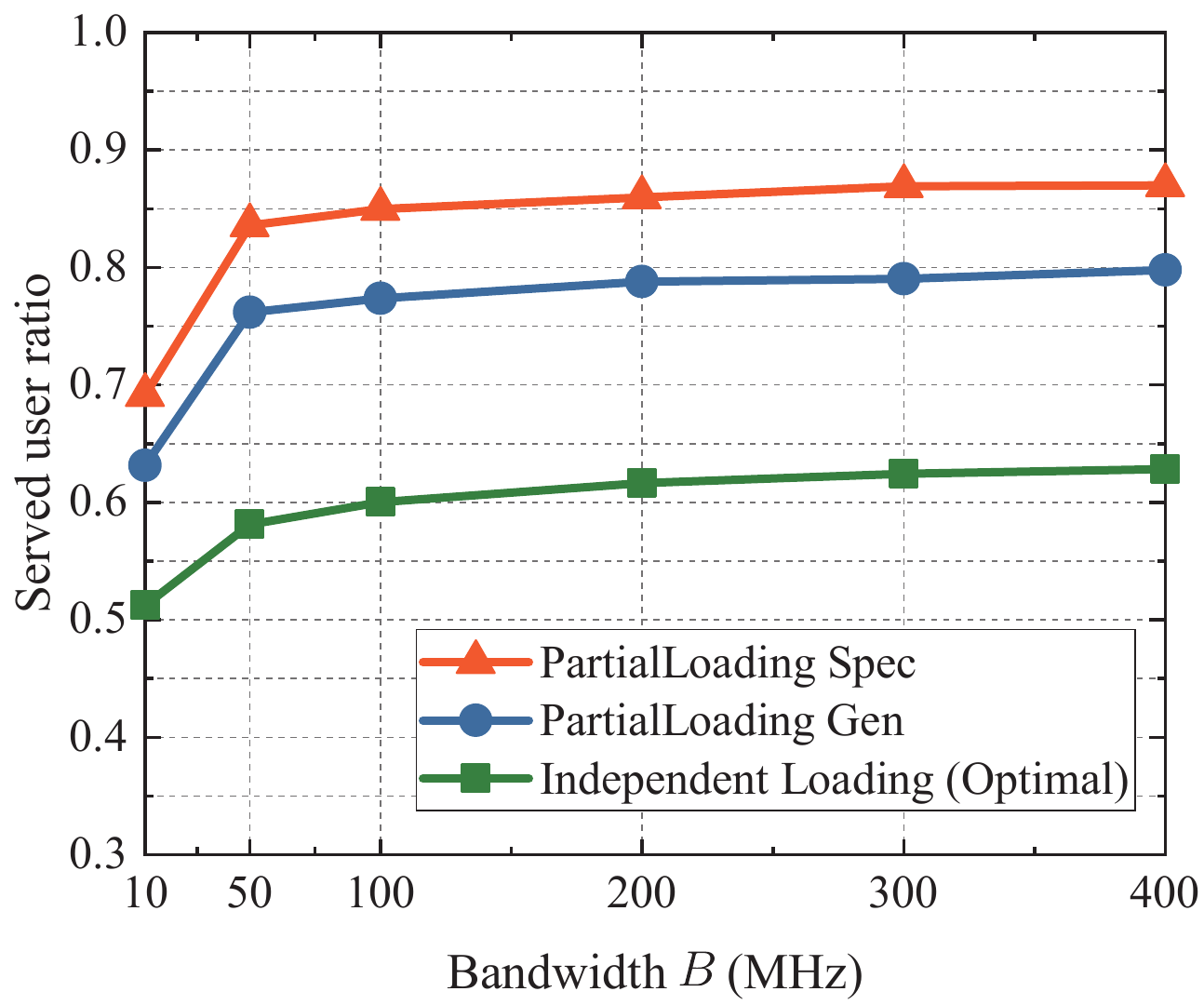}\label{fig:spec_bw}}
	\quad
	\subfigure[Served user ratio vs. $K$.]{\includegraphics[height=2.9cm, keepaspectratio]{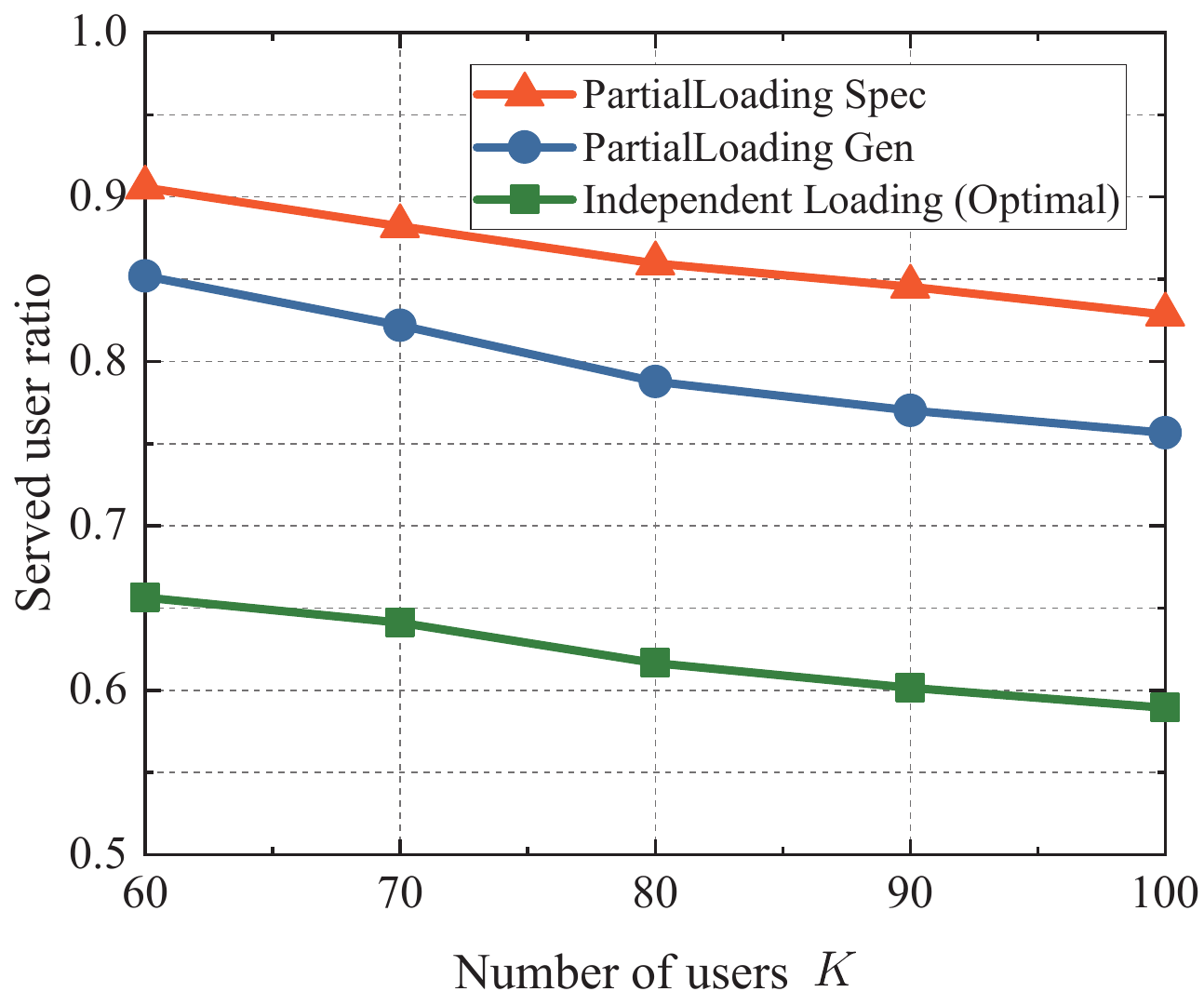}\label{fig:spec_user}}
        \\
        \subfigure[Served user ratio vs. $\bar{T}$.]{\includegraphics[height=2.9cm, keepaspectratio]{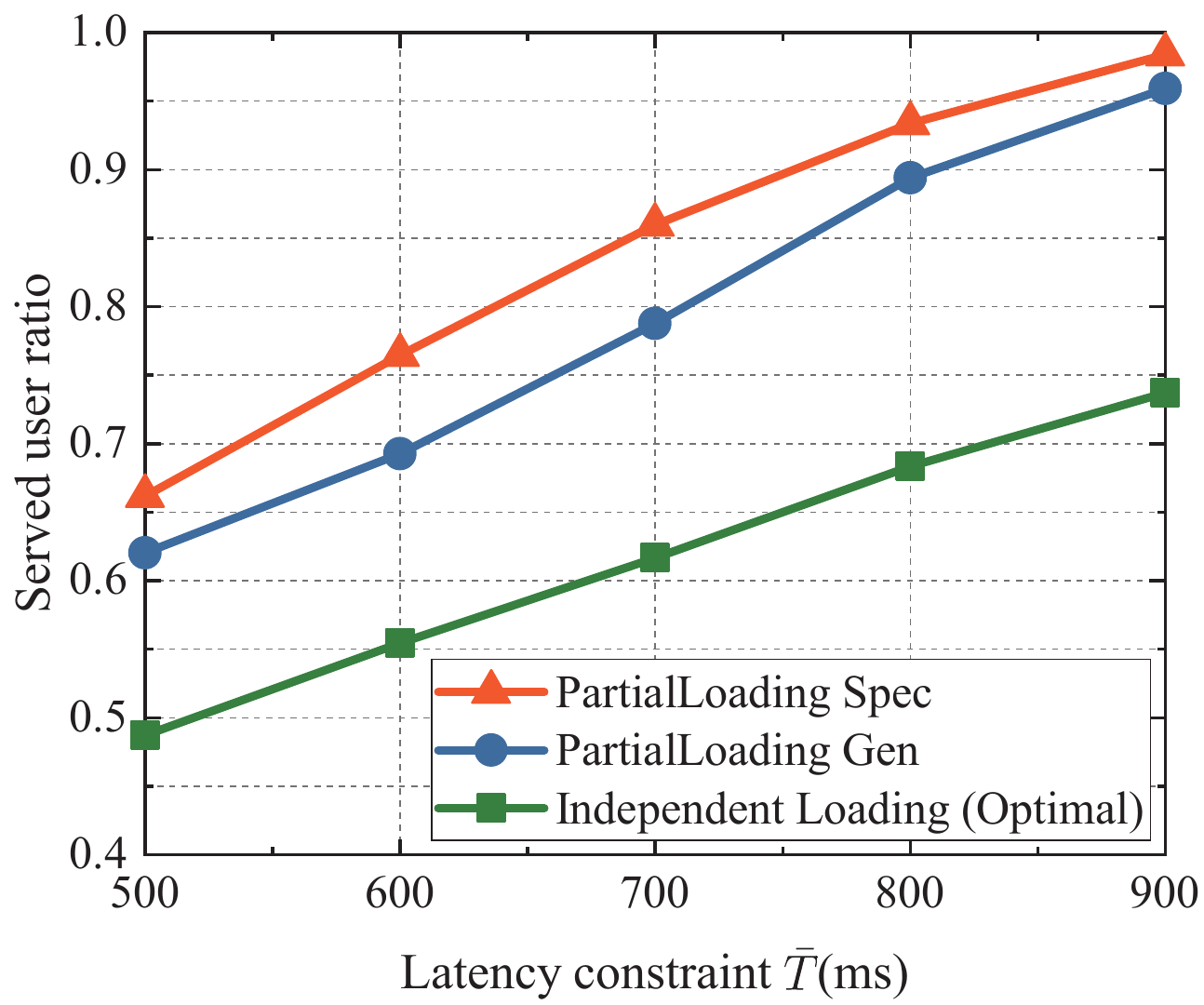}\label{fig:spec_ddl}}
	\quad
	\subfigure[Served user ratio vs. $\theta$.]{\includegraphics[height=2.9cm, keepaspectratio]{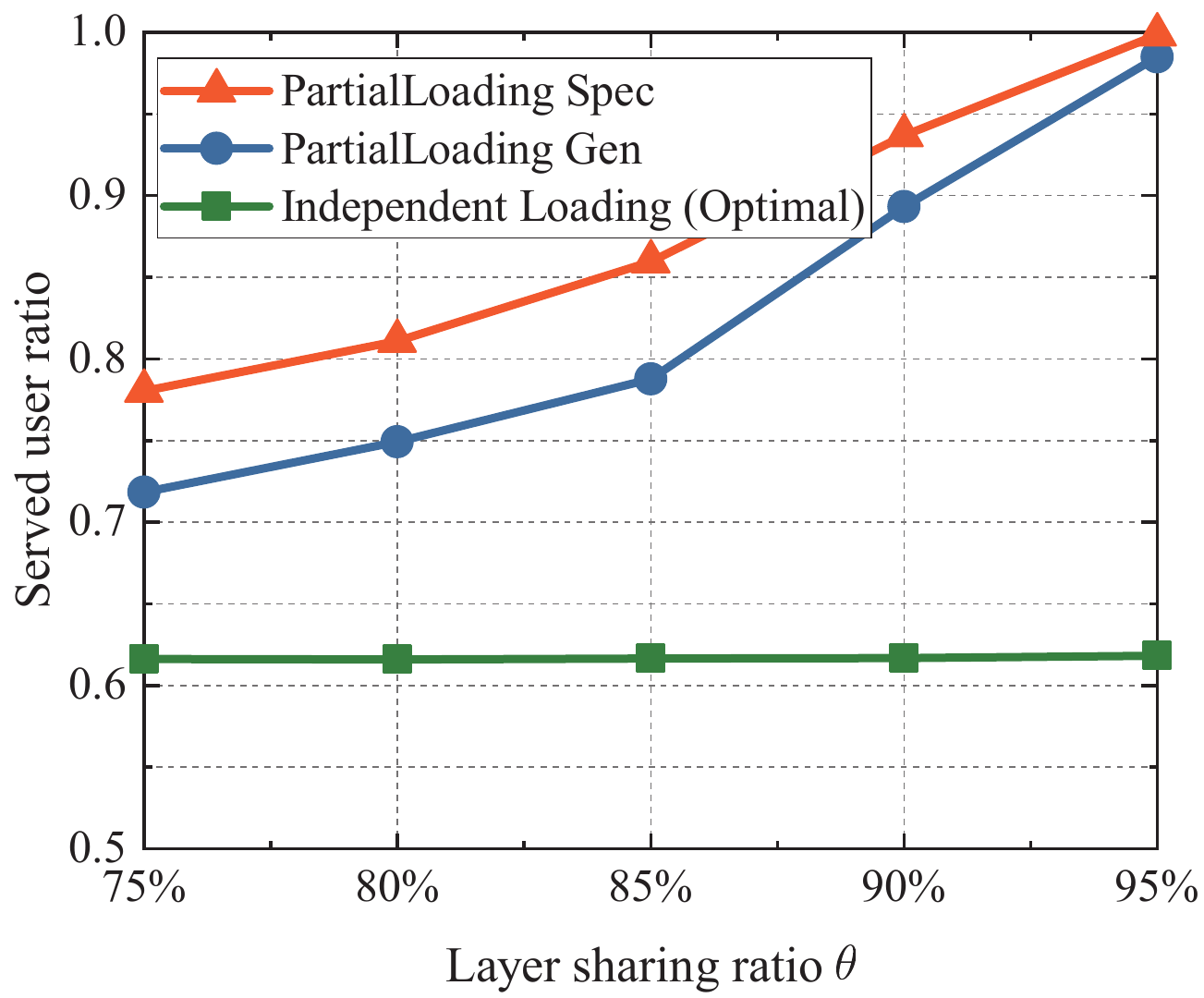}\label{fig:spec_ratio}}
    \vspace{-3pt}
 \caption{Served user ratio in the BS case, where $B$, $K$, $\bar{T}$, and $\theta$ are respectively set to 200 MHz, 80, 700 ms, and 85\% by default.}
 \vspace{-10pt}\label{fig:special}
\end{figure}

Fig. \ref{fig:spec_user} illustrates the served user ratio for varying $K$ in the edge network. As $K$ increases, the served user ratio decreases due to the communication and computing resource limitations of the edge server, indicating that the system’s capacity has been reached. Nevertheless, the edge server can still serve nearly 83\% and 76\% of users with PartialLoading Spec and PartialLoading Gen, respectively. This demonstrates that the proposed approaches still perform well even with a high density of users. Moreover, PartialLoading Spec and PartialLoading Gen outperform Independent Loading, improving the served user ratio by about 24\% and 18\%, respectively.

In Fig. \ref{fig:spec_ddl}, we study the impact of $\bar{T}$ on the performance. As expected, the served user ratio increases proportionally with $\bar{T}$. This figure also shows a nearly linear relationship between the served user ratio and $\bar{T}$. Furthermore, PartialLoading Spec and PartialLoading Gen improve the served user ratio by about 22\% and 18\%, respectively, over Independent Loading.

We finally analyze the effect of $\theta$ 
in Fig. \ref{fig:spec_ratio}. PartialLoading Spec and PartialLoading Gen enhance the served user ratio by about 26\% and 21\% than Independent Loading. Besides, the served user ratio increases at a faster rate as $\theta$ grows. The reason is that the size of the top layers in the utilized models is generally larger than that of the bottom layers. In the BS case, a higher $\theta$ implies that more top layers are shared and do not need to be loaded. Therefore, the increasing trend of the served user ratio becomes more pronounced.

\subsection{Performance in the General Case}
In this subsection, we evaluate the performance of PartialLoading Gen in the general case, where the shared parameter blocks can appear at arbitrary positions within AI models. Since PartialLoading Spec is specifically designed for the BS case, in this subsection, we only compare PartialLoading Gen with Independent Loading.

Similar to the above subsection, we analyze the served user ratio of PartialLoading Gen under varying values of $B$, $K$, $\bar{T}$, and $\theta$ in Fig. \ref{fig:general}. In Figs. \ref{fig:general_bw}, \ref{fig:general_user}, \ref{fig:general_ddl}, and \ref{fig:general_ratio}, trends similar to those in Fig. \ref{fig:special} can be observed. In these figures, PartialLoading Gen improves the served user ratio by about 36\%, 38\%, 36\%, and 39\%, respectively, compared with Independent Loading. Moreover, in Fig. \ref{fig:general_ratio}, due to the arbitrary position of shared parameter blocks, the increasing trend of the served user ratio remains stable as $\theta$ grows.
\begin{figure}[t]
    \centering
	\subfigure[Served user ratio vs. $B$.]{\includegraphics[height=2.9cm, keepaspectratio]{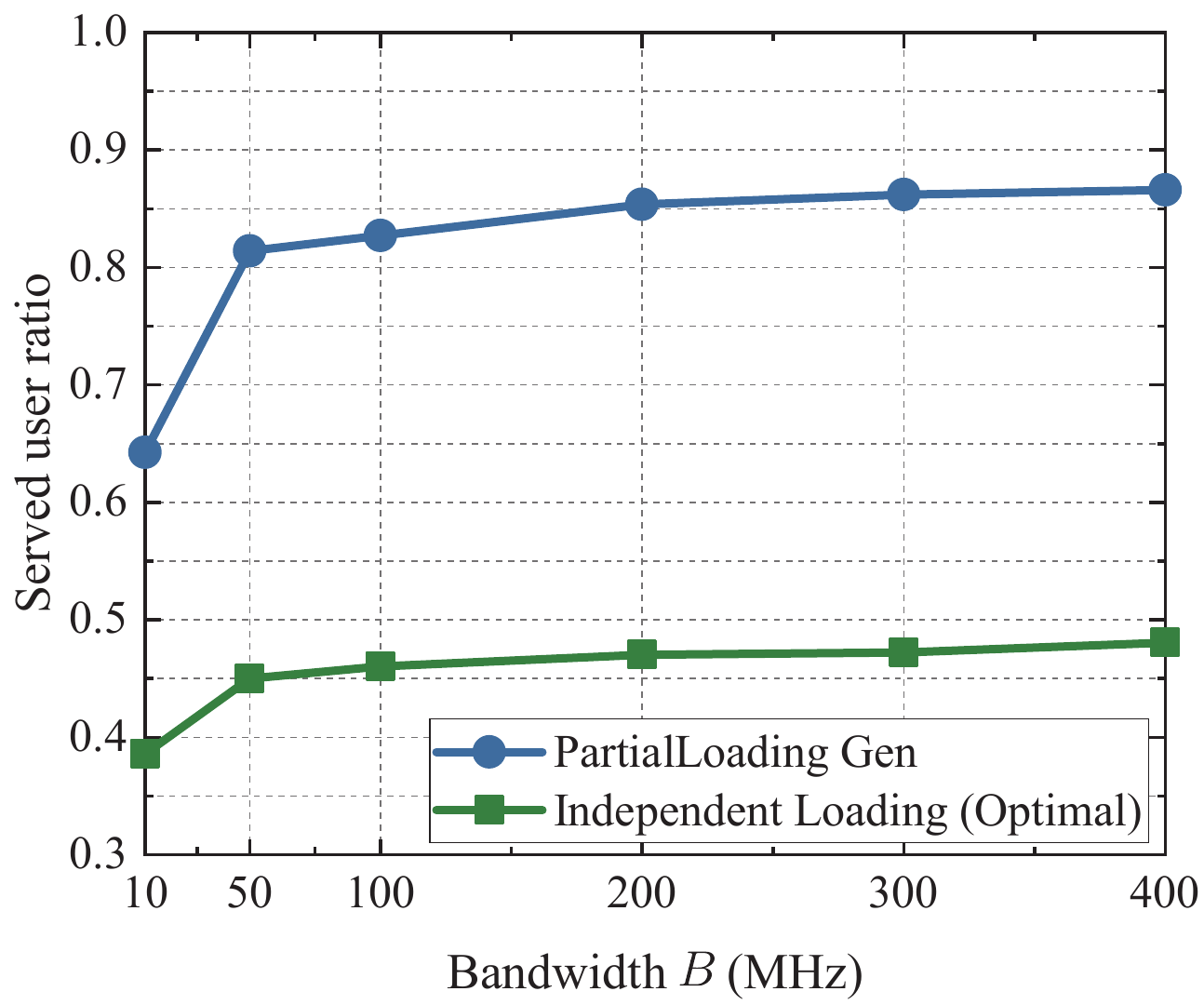}\label{fig:general_bw}}
	\quad
	\subfigure[Served user ratio vs. $K$.]{\includegraphics[height=2.9cm, keepaspectratio]{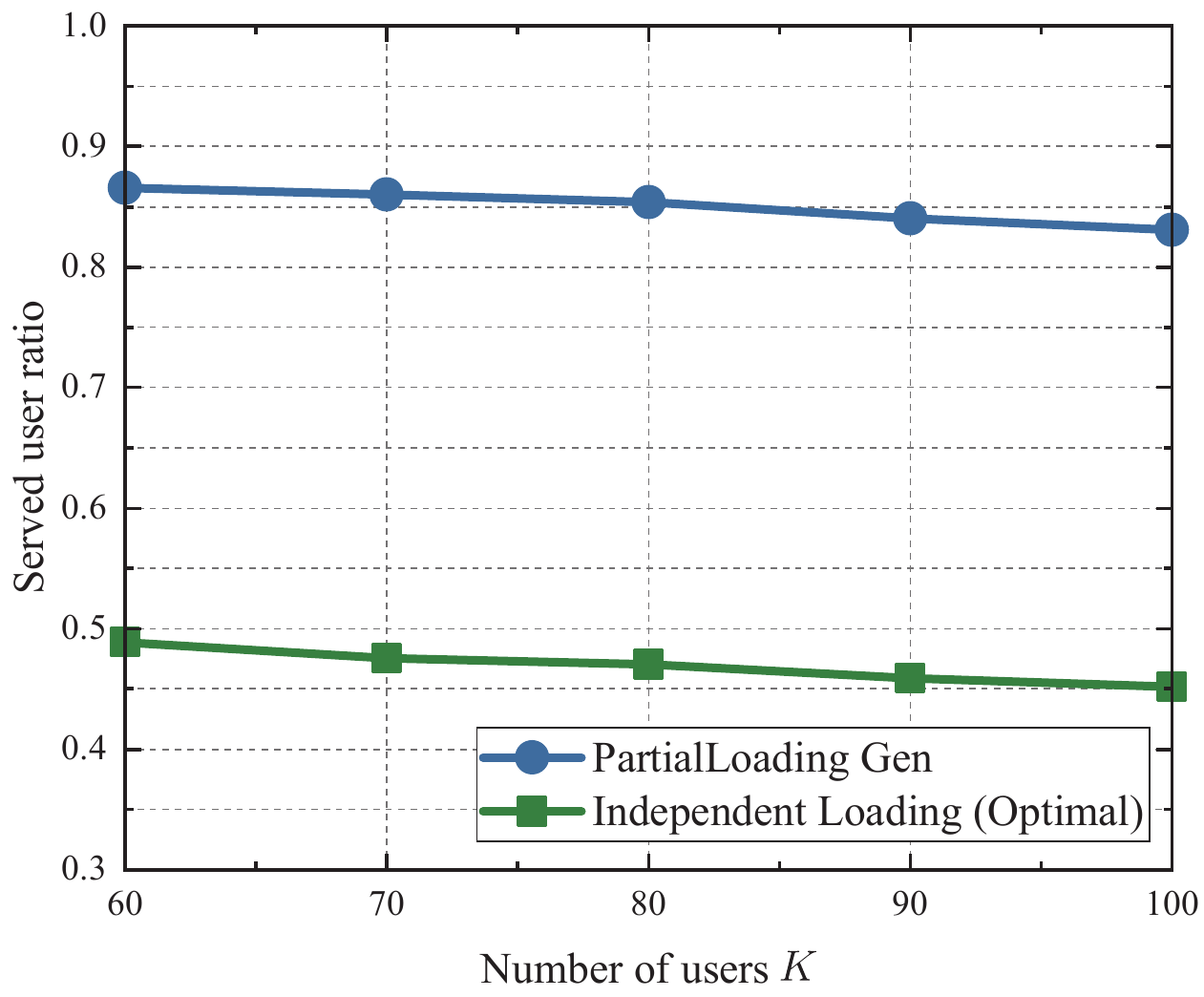}\label{fig:general_user}}
        \\
        \subfigure[Served user ratio vs. $\bar{T}$.]{\includegraphics[height=2.9cm, keepaspectratio]{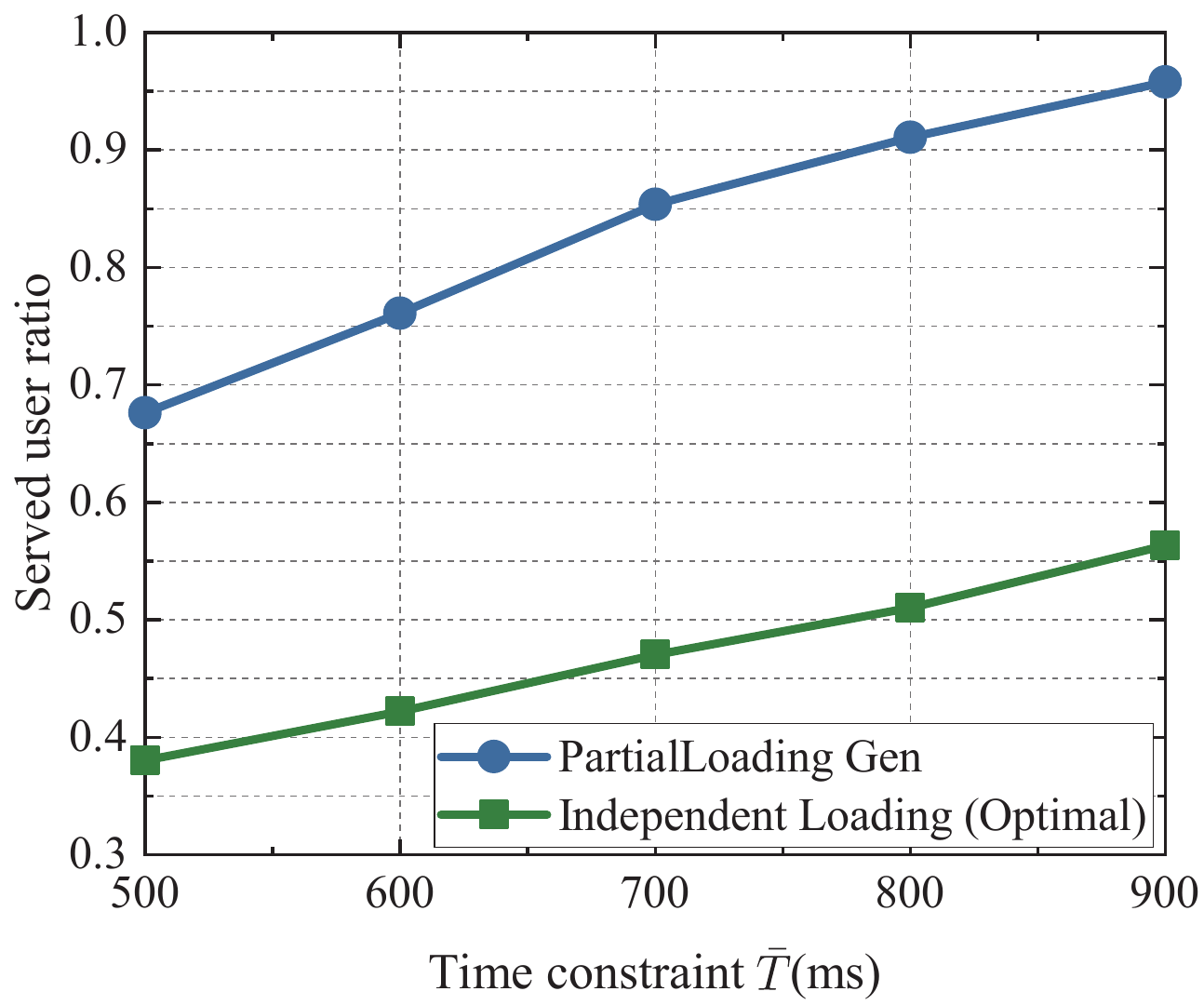}\label{fig:general_ddl}}
	\quad
	\subfigure[Served user ratio vs. $\theta$.]{\includegraphics[height=2.9cm, keepaspectratio]{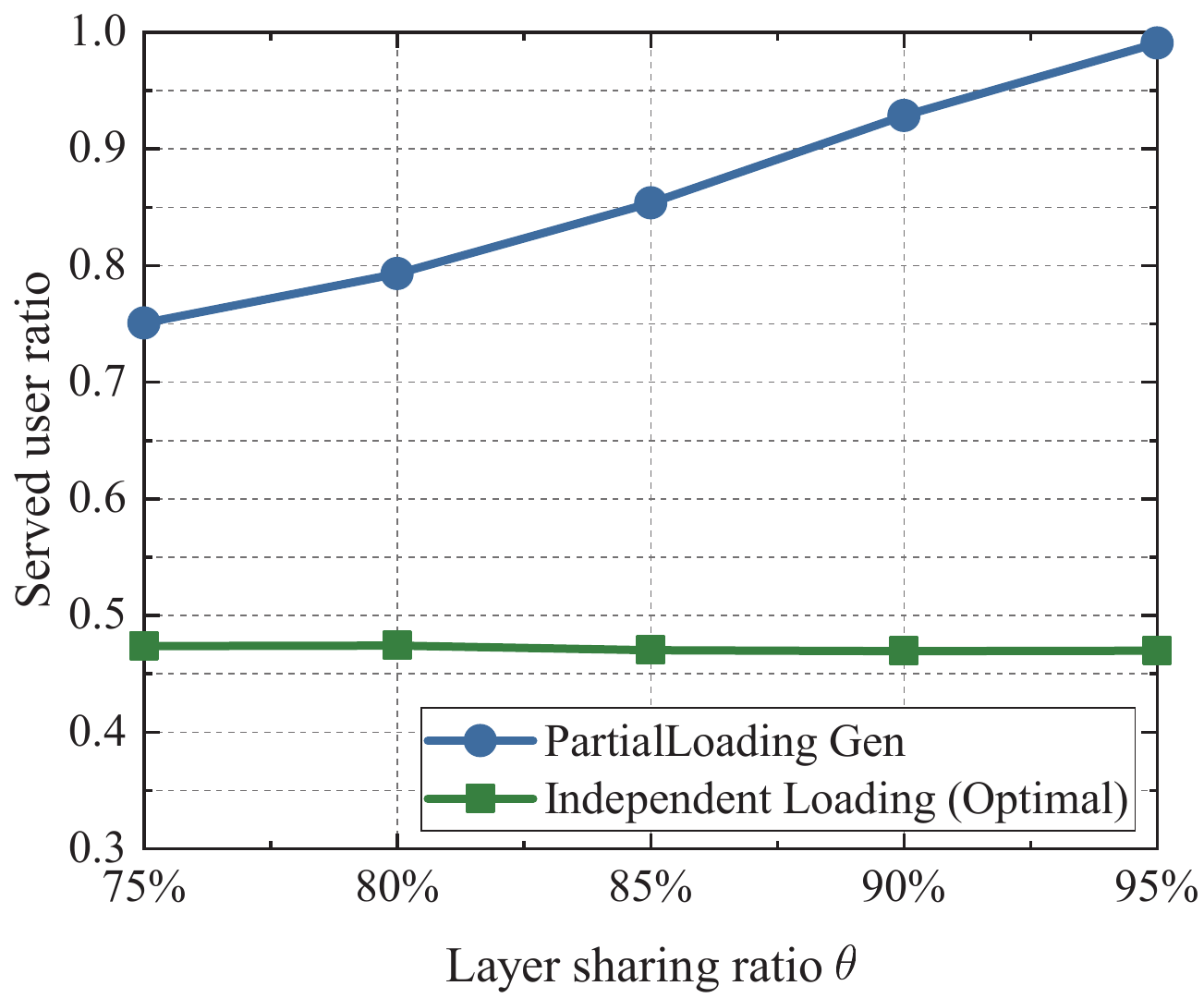}\label{fig:general_ratio}}
    \vspace{-3pt}
 \caption{Served user ratio in the general case, where the default values of $B$, $K$, $\bar{T}$, and $\theta$ are the same as those in Fig. \ref{fig:special}.}
 \vspace{-10pt}\label{fig:general}
\end{figure}

\subsection{Ablation Study}
This subsection evaluates the impact of spectrum bandwidth allocation on the served user ratio in the proposed framework through ablation experiments. Fig. \ref{fig:ablation} compares the performance of the proposed PartialLoading Spec and PartialLoading Gen algorithms with their counterparts under equal bandwidth allocation in both the BLS and general cases. In the setup of algorithms with equal bandwidth allocation, the total bandwidth $B$ is evenly divided into $r\in\{5, 10, 20\}$ sub-channels, with each sub-channel assigned to a single user. The bandwidth allocated to user $k$ in batch $n$ in \eqref{eq_bw} is updated to $B'_{n,k} = x_{n,k}B_{k}$, where $B_{k}=\frac{B}{r}$, satisfying $\sum\limits_{k\in\mathcal{K}}x_{n,k}B_{k} \le B, \ \forall n\in\mathcal{N}$. Then, the corresponding user scheduling is determined using $B'_{n,k}$ for the BS case and general case by PartialLoading Spec and PartialLoading Gen, respectively. 

Figs. \ref{fig:spec_bw_ablation} and \ref{fig:spec_user_ablation} demonstrate that the proposed PartialLoading Spec consistently outperforms the algorithms with equal bandwidth allocation across varying $B$ and $K$ in the BS case. Specifically, in Fig. \ref{fig:spec_bw_ablation}, PartialLoading Spec achieves average served user ratio improvements of approximately 9.9\%, 16.9\%, and 26.2\% over the equal-bandwidth-allocation algorithms with $r=5$, $r=10$, and $r=20$, respectively. Similarly, in Fig. \ref{fig:spec_user_ablation}, the respective improvements are approximately 3.9\%, 7.7\%, and 15.2\%. 

Figs. \ref{fig:general_bw_ablation} and \ref{fig:general_user_ablation} exhibit trends similar to those in Figs. \ref{fig:spec_bw_ablation} and \ref{fig:spec_user_ablation}, demonstrating PartialLoading Gen consistently outperforms the equal-bandwidth-allocation algorithms across varying $B$ and $K$ in the general case. In Fig. \ref{fig:general_bw_ablation}, the respective improvements are about 9.6\%, 16.1\%, and 25.6\%, while in Fig. \ref{fig:general_user_ablation}, they are about 3.4\%, 7.7\%, and 16.1\%.

\begin{figure}[t]
    \centering
	\subfigure[Served user ratio vs. $B$ in the BS case.]{\includegraphics[height=2.9cm, keepaspectratio]{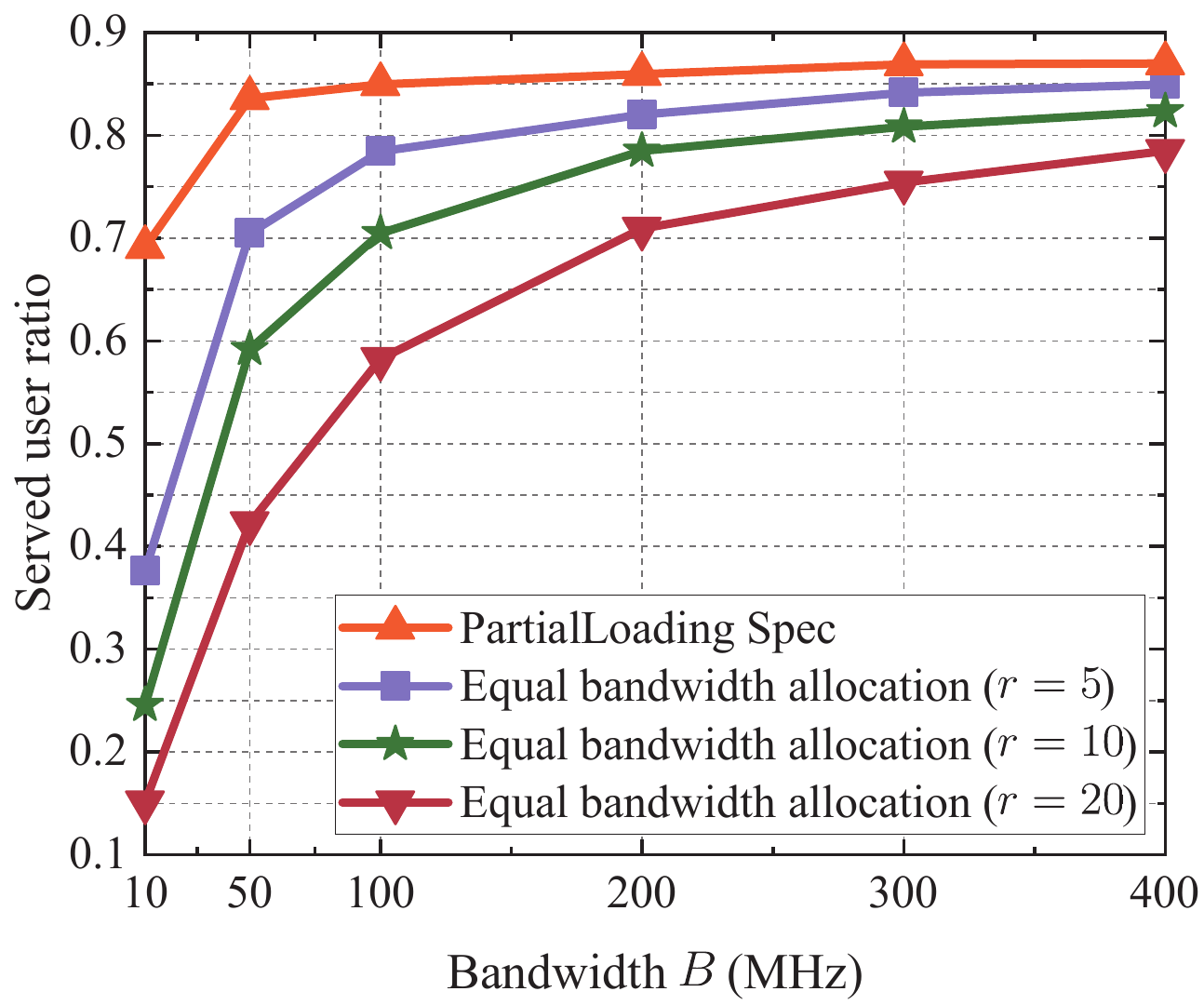}\label{fig:spec_bw_ablation}}
	\quad
	\subfigure[Served user ratio vs. $K$ in the BS case.]{\includegraphics[height=2.9cm, keepaspectratio]{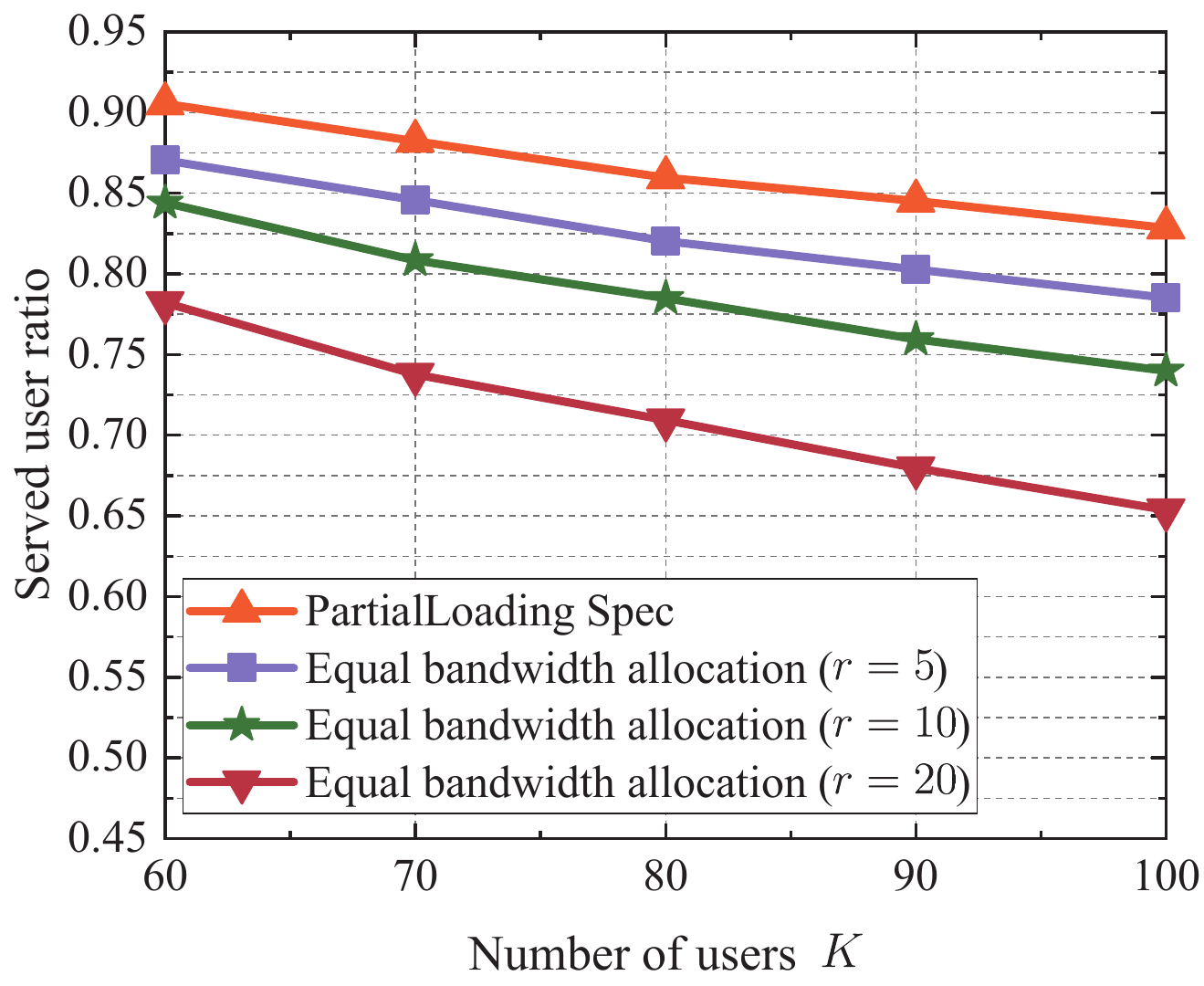}\label{fig:spec_user_ablation}}
\\
        \subfigure[Served user ratio vs. $B$ in the general case.]{\includegraphics[height=2.9cm, keepaspectratio]{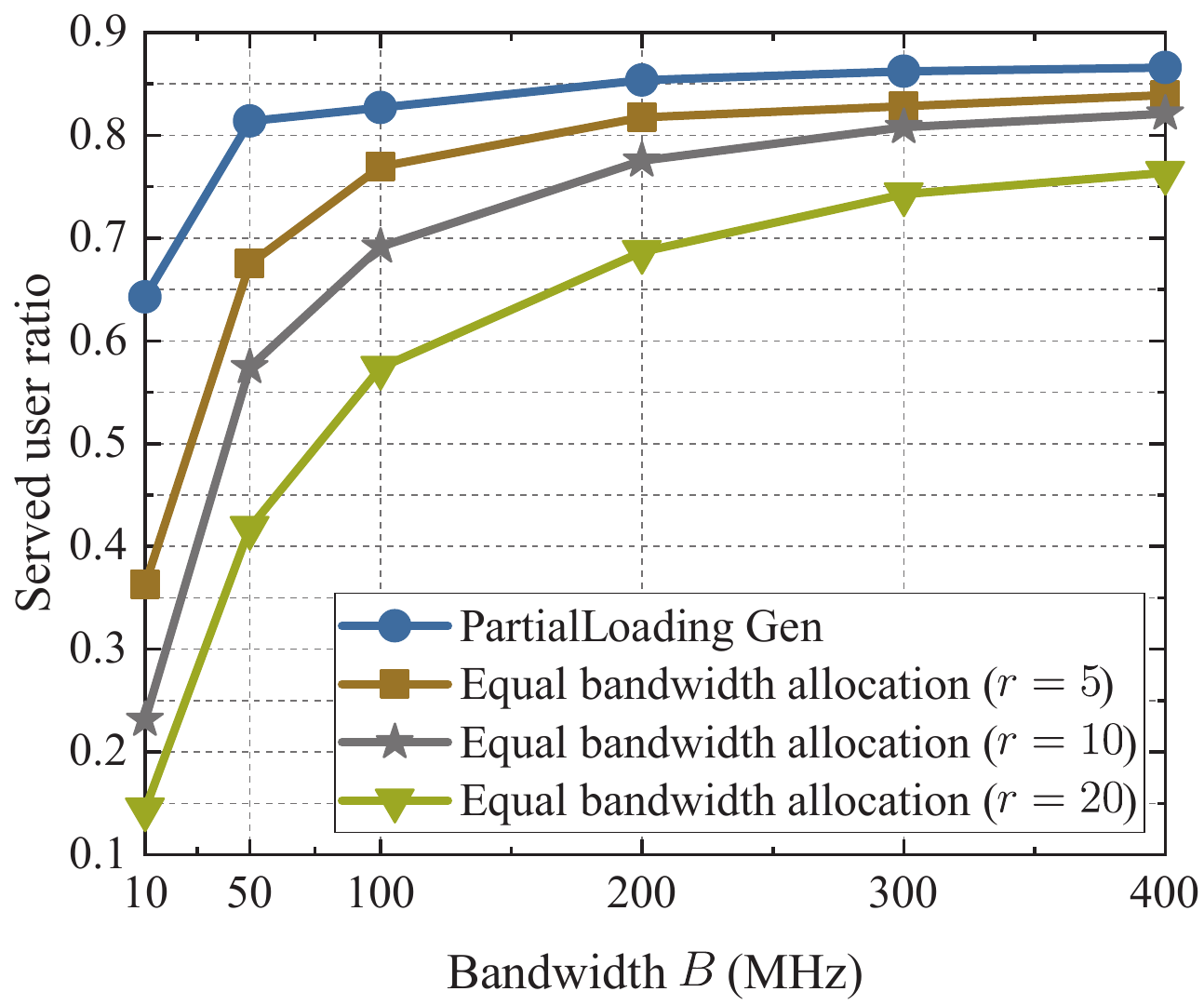}\label{fig:general_bw_ablation}}
	\quad
	\subfigure[Served user ratio vs. $K$ in the general case.]{\includegraphics[height=2.9cm, keepaspectratio]{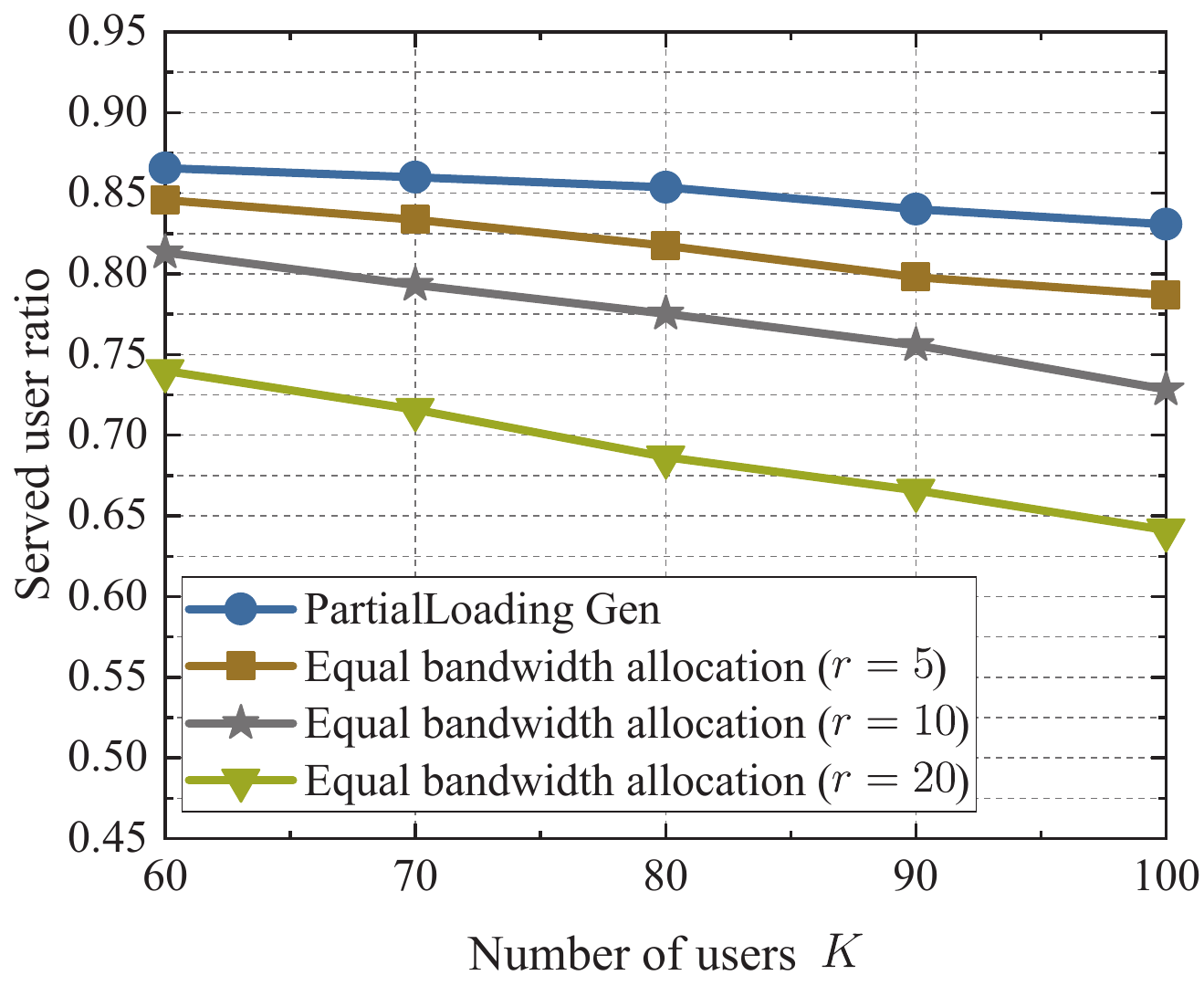}\label{fig:general_user_ablation}}
    \vspace{-3pt}
 \caption{Ablation study for bandwidth allocation, where the default values of $B$, $K$, $\bar{T}$, and $\theta$ are the same as those in Fig. \ref{fig:special}. In the equal-bandwidth-allocation algorithms, $r$ represents the number of sub-channels into which the total bandwidth $B$ is evenly divided, with each sub-channel assigned to a single user.
 }
 \vspace{-5pt}\label{fig:ablation}
\end{figure}

\subsection{Running Time Comparisons}
In this subsection, we evaluate the efficiency of the proposed algorithms by comparing the served user ratio and average running time of the proposed algorithms against the exhaustive search. To ensure the feasible execution of the exhaustive search, we consider a small problem scale with $K$, $I$, and $\bar{T}$ set to 20, 5, and 200 ms, respectively. 

Fig. \ref{fig:spec_optimal} compares the performance among PartialLoading Gen, PartialLoading Spec, and the exhaustive search in the BS case. It illustrates that the served user ratio of PartialLoading Spec matches that of the exhaustive search, validating the optimality of the PartialLoading Spec. However, the running time of PartialLoading Spec is only about 5.6 ms, making it nearly 3,681 times faster than the exhaustive search, where the time complexity of the exhaustive search grows exponentially with $I$ and $\bar{T}$. Additionally, the served user ratio of PartialLoading Gen is only about 3.7\% lower than the exhaustive search, yet it achieves approximately 3,749 times faster execution. 

Similarly, in Fig. \ref{fig:general_optimal}, we compare the performance between PartialLoading Gen and the exhaustive search in the general case. The served user ratio of PartialLoading Gen only degrades by about 4.3\%. In contrast, it runs in only about 12.5 ms and accelerates the decision process by about 2,478 times.

\begin{figure}[t]
    \centering
	\subfigure[
    Results in the BS case.]{\includegraphics[height=2.8cm, keepaspectratio]{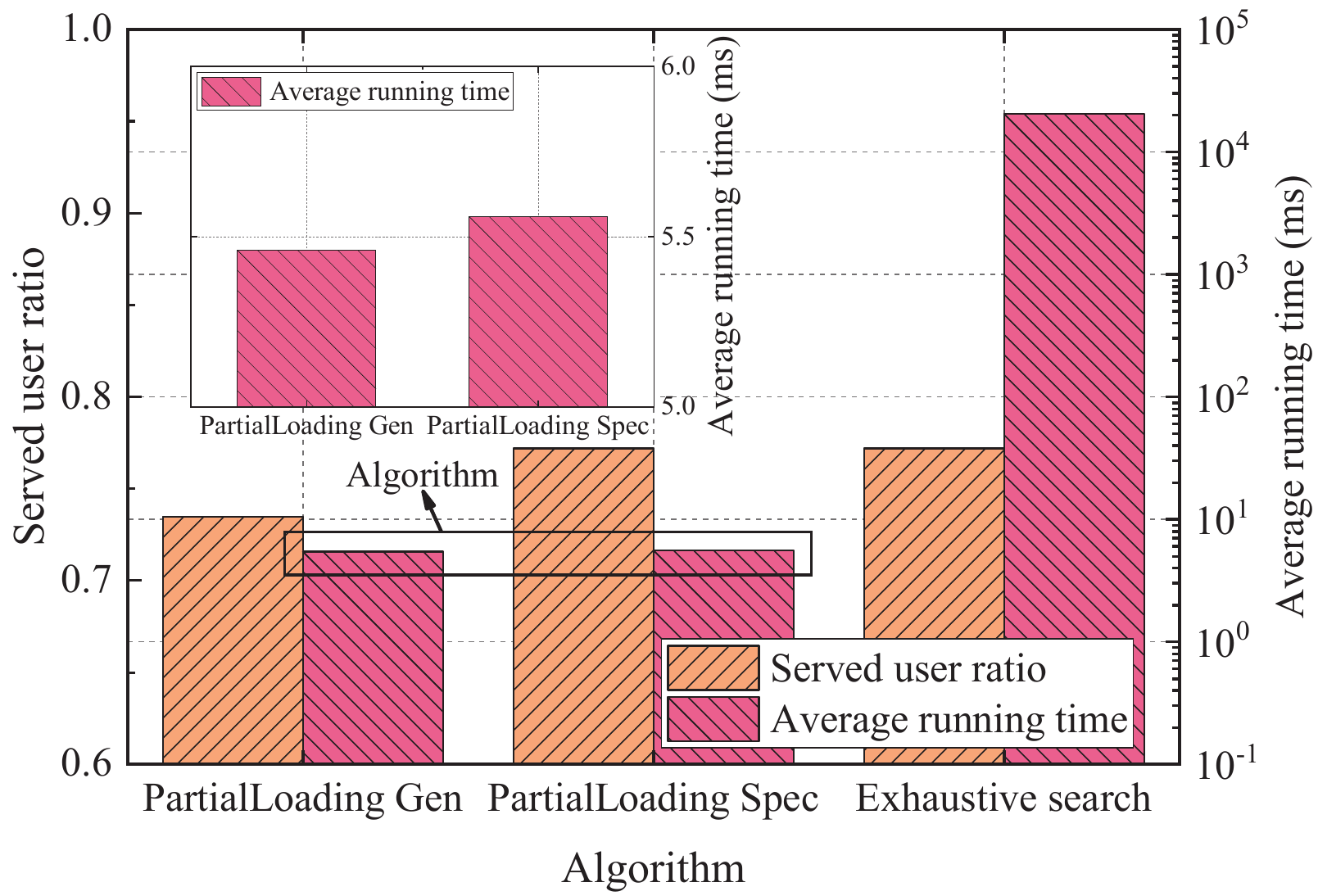}\label{fig:spec_optimal}}
	\quad
	\subfigure[
    Results in the general case.]{\includegraphics[height=2.8cm, keepaspectratio]{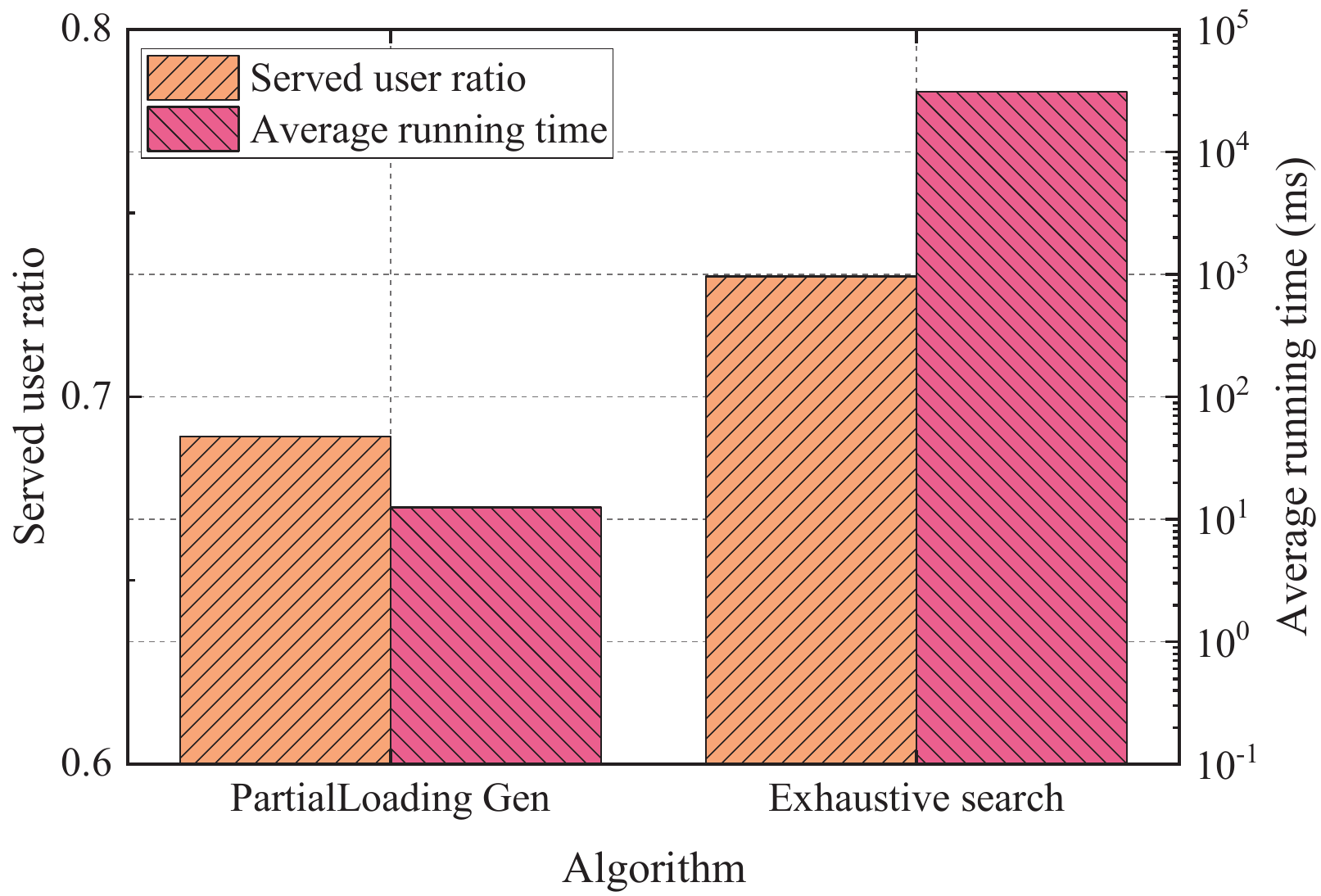}\label{fig:general_optimal}}
    \vspace{-3pt}
 \caption{Comparisons of served user ratio and average running time for different algorithms, with default values of $B$, and $\theta$ as in Fig. \ref{fig:special}.}
 \vspace{-10pt}\label{fig:optimal}
\end{figure}
\section{Conclusions}
In this paper, we have investigated the multi-user scheduling problem in parameter-sharing edge inference and proposed a PartialLoading framework to maximize inference task throughput. By leveraging shared parameter blocks to avoid redundant model loading across consecutive inference batches, the proposed framework optimizes task throughput by scheduling users and allocating spectrum bandwidth under latency and resource constraints. To simplify the solution approach, we have decomposed the original problem into two sub-problems, i.e., user scheduling and bandwidth allocation sub-problems, without compromising the optimality of the original problem. For the bandwidth allocation sub-problem, we have derived the optimal closed-form expression for spectrum bandwidth allocation given user scheduling decisions. For the user-scheduling sub-problem, we first investigate the backbone-sharing case and design a dynamic programming-based algorithm to find the optimal solution in polynomial time. Then, we have addressed the general case by developing a greedy heuristic procedure. Both the proposed algorithms have demonstrated significant task throughput improvements compared with traditional frameworks without exploiting parameter-sharing. By pioneering user scheduling for parameter-sharing edge inference, we hope this work offers insights into further improving the efficiency of edge inference systems.

\bibliographystyle{IEEEtran}
\bibliography{IEEEabrv,reference}
\newpage
\clearpage
\setcounter{page}{1}

\appendix
Before presenting the detailed proof, we introduce the following statements, which apply to Appendices \ref{proof:proposition_2_1} through \ref{proof:proposition_3_3}. A user scheduling sequence ${\bf{X}} = \left\{{\bf{X}}_{1},\dots,{\bf{X}}_{N}\right\}$ can be partitioned into multiple disjoint sub-sequences. Denote the $a$-th sub-sequence as $\Pi_{a} = \left\{{\bf{X}}_{n},{\bf{X}}_{n+1},\dots,{\bf{X}}_{n'}\right\}$, where $1\le n\le n'\le N$. Moreover, let $c\left(\Pi_{a}\right)$ be the completion time of $\Pi_{a}=\left\{\dots, {\bf{X}}_{n}\right\}$, where $c\left(\Pi_{a}\right)=\sum\limits_{n'=1}^{n}t_{n'}$.
\begin{appendices}
\section{Proof of Proposition \ref{proposition_1}}\label{proof:proposition_1}
We first prove the first statement by showing that minimizing $t_{n}^{\text{up}}$ is a sufficient condition for optimal $y_{n,k}$. Given a user scheduling ${\bf{X}}$, suppose that minimizing $t_{n}^{\text{up}}$ leads to a bandwidth allocation ${\bf{Y}}$, which is not optimal under ${\bf{X}}$. This implies that the total uploading time of the users in ${\bf{X}}$ with the optimal bandwidth allocation is no shorter than that with ${\bf{Y}}$. However, this contradicts the definition of the optimal bandwidth allocation. This is because ${\bf{Y}}$ yields the minimum total uploading time for the users in ${\bf{X}}$, and the total model loading and inference computing time are independent of the bandwidth allocation, implying that allocating bandwidth according to ${\bf{Y}}$ serves at least the same number of users as the optimal bandwidth allocation. Therefore, minimizing $t_{n}^{\text{up}}$ guarantees the optimality of the bandwidth allocation for the users in ${\bf{X}}$, completing the proof of the first statement. 
    
    Now we prove the second statement and derive \eqref{optimal_uploading} and \eqref{optimal_bandwidth}. Let $\dot{\mathcal{K}}_{n}=\left\{k\mid x_{n,k}=1\right\}$ be the set of scheduled users in batch $n$. 
    
    If $\left|\dot{\mathcal{K}}_{n}\right|\ge1$, then $t_{n}^{\text{up}}
    =\mathop{\max}\limits_{k\in\mathcal{K}}\left\{t_{n,k}^{\text{up}}\right\}
    =\mathop{\max}\limits_{k\in\dot{\mathcal{K}}_{n}}\left\{t_{n,k}^{\text{up}}\right\}
    \ge\frac{\sum\limits_{k\in\dot{\mathcal{K}}_{n}}t_{n,k}^{\text{up}}}{\left|{\dot{\mathcal{K}}_{n}}\right|}
    \ge \sqrt[\left|{\dot{\mathcal{K}}_{n}}\right|]{\prod\limits_{k\in\dot{\mathcal{K}}_{n}}t_{n,k}^{\text{up}}}$. The second inequality holds with equality if and only if $t_{n,k}^{\text{up}} = t_{n,k'}^{\text{up}}$ for all $k, k' \in \dot{\mathcal{K}}_{n}$, i.e., all users in $\dot{\mathcal{K}}_{n}$ have the same uploading time.
    On the one hand, for all $k \in \dot{\mathcal{K}}_{n}$, since $t_{n,k}^{\text{up}} = \frac{D_{k}}{y_{n,k}B\bar{R}_{k}}$, the equality $\frac{\sum\limits_{k\in\dot{\mathcal{K}}_{n}}t_{n,k}^{\text{up}}}{\left|{\dot{\mathcal{K}}_{n}}\right|}
    =\sqrt[\left|{\dot{\mathcal{K}}_{n}}\right|]{\prod\limits_{k\in\dot{\mathcal{K}}_{n}}t_{n,k}^{\text{up}}}$ holds if and only if $y_{n,k} = \frac{D_{k}}{B\bar{R}_{k}\sum\limits_{k'\in\dot{\mathcal{K}}_{n}}\frac{D_{k'}}{B\bar{R}_{k'}}}= \frac{D_{k}}{B\bar{R}_{k}\sum\limits_{k'\in\mathcal{K}}\frac{x_{n,k'}D_{k'}}{B\bar{R}_{k'}}}$. Substituting this into $t_{n,k}^{\text{up}}$, the minimum uploading time is given by $t_{n,\min}^{\text{up}}=\sqrt[\left|{\dot{\mathcal{K}}_{n}}\right|]{\prod\limits_{k\in\dot{\mathcal{K}}_{n}}t_{n,k}^{\text{up}}}=t_{n,k}^{\text{up}}= \frac{D_{k}}{\frac{D_{k}}{B\bar{R}_{k}\sum\limits_{k'\in\mathcal{K}}\frac{x_{n,k'}D_{k'}}{B\bar{R}_{k'}}}B\bar{R}_{k}}=\sum\limits_{k'\in\mathcal{K}}\frac{x_{n,k'}D_{k'}}{B\bar{R}_{k'}}$.
    On the other hand, for all $k\in \mathcal{K}\setminus\dot{\mathcal{K}}_{n}$, since $x_{n,k} = 0$, we have $y_{n,k} = 0$, which is constrained by \eqref{const_7}. 
    
    If $\left|\dot{\mathcal{K}}_{n}\right|=0$, it indicates that no users are scheduled in batch $n$, i.e., $x_{n,k}=0, \ \forall k\in\mathcal{K}$. Then, $t_{n}^{\text{up}}=0$ and $y_{n,k}=0$, which are consistent with \eqref{optimal_uploading} and \eqref{optimal_bandwidth}, respectively. This completes the proof of the second statement and concludes the proof.
\section{Proof of Theorem \ref{proposition_2_1}}\label{proof:proposition_2_1}
    The proof of the first statement proceeds as follows. Suppose the optimal user scheduling is $\left\{\dots,\Pi_{a},\Pi_{a+1},\Pi_{a+2},\Pi_{a+3},\dots\right\}$, where users in $\Pi_{a}$ and $\Pi_{a+2}$ request the same AI model $i$, users in $\Pi_{a+3}$ request model $i_{3}$, and users in $\Pi_{a+1}=\left\{{\bf{X}}_{n_{1}},\dots,{\bf{X}}_{n_{2}}\right\}$ request at least one AI model other than model $i$. Moreover, let models $i_{1}$ and $i_{2}$ denote the models requested in batches $n_{1}$ and $n_{2}$, respectively, where $i_{3}\ne i$, $i_{1}\ne i$, and $i_{2}\ne i$. Now consider exchanging $\Pi_{a+1}$ and $\Pi_{a+2}$ in the original user scheduling. Denote the completion time of $\Pi_{a+3}$ in the new user scheduling as $c'\left(\Pi_{a+3}\right)$. We claim that $c'\left(\Pi_{a+3}\right)\le c\left(\Pi_{a+3}\right)$. The reasoning is as follows. In the original user scheduling $\left\{\dots,\Pi_{a},\Pi_{a+1},\Pi_{a+2},\Pi_{a+3},\dots\right\}$, the data size of parameter blocks that need to be loaded in the first batch of $\Pi_{a+1}$, $\Pi_{a+2}$, and $\Pi_{a+3}$ are, respectively: $\sum\limits_{j\in \mathcal{J}_{i_{1}}\setminus\mathcal{J}_{i}}S_{j}$, $\sum\limits_{j\in \mathcal{J}_{i}\setminus\mathcal{J}_{i_{2}}}S_{j}$, and $\sum\limits_{j\in \mathcal{J}_{i_{3}}\setminus\mathcal{J}_{i}}S_{j}$. 
    In the new user scheduling $\left\{\dots,\Pi_{a},\Pi_{a+2},\Pi_{a+1},\Pi_{a+3},\dots\right\}$, the data size of parameter blocks that need to be loaded in the first batch of $\Pi_{a+2}$, $\Pi_{a+1}$, and $\Pi_{a+3}$ are, respectively: 0, $\sum\limits_{j\in \mathcal{J}_{i_{1}}\setminus\mathcal{J}_{i}}S_{j}$, and $\sum\limits_{j\in \mathcal{J}_{i_{3}}\setminus\mathcal{J}_{i_{2}}}S_{j}$. For any parameter block $ j\in \mathcal{J}_{i_{3}}\setminus\mathcal{J}_{i_{2}}$, if $j\in\mathcal{J}_{i}$, then $j\in\mathcal{J}_{i}\setminus\mathcal{J}_{i_{2}}$, as $j\notin\mathcal{J}_{i_{2}}$. Conversely, if $j\notin\mathcal{J}_{i}$, then $j\in\mathcal{J}_{i_{3}}\setminus\mathcal{J}_{i}$, as $j\in\mathcal{J}_{i_{3}}$. Thus, for all $j\in\mathcal{J}_{i_{3}}\setminus\mathcal{J}_{i_{2}}$, we have $j\in \left(\mathcal{J}_{i}\setminus\mathcal{J}_{i_{2}}\right)\cup\left(\mathcal{J}_{i_{3}}\setminus\mathcal{J}_{i}\right)$. This implies that $\mathcal{J}_{i_{3}}\setminus\mathcal{J}_{i_{2}}\subseteq\left(\mathcal{J}_{i}\setminus\mathcal{J}_{i_{2}}\right)\cup\left(\mathcal{J}_{i_{3}}\setminus\mathcal{J}_{i}\right)$. Consequently, we have: $\sum\limits_{j\in \mathcal{J}_{i_{1}}\setminus\mathcal{J}_{i}}S_{j}+\sum\limits_{j\in \mathcal{J}_{i}\setminus\mathcal{J}_{i_{2}}}S_{j}+\sum\limits_{j\in \mathcal{J}_{i_{3}}\setminus\mathcal{J}_{i}}S_{j}\ge\sum\limits_{j\in \mathcal{J}_{i_{1}}\setminus\mathcal{J}_{i}}S_{j}+\sum\limits_{j\in \mathcal{J}_{i_{3}}\setminus\mathcal{J}_{i_{2}}}S_{j}$. This inequality represents that fewer parameter blocks need to be loaded in the new user scheduling after exchanging $\Pi_{a+1}$ and $\Pi_{a+2}$ in the original user scheduling. As a result, the completion time $c'(\Pi_{a+3})$ of $\Pi_{a+3}$ in the new scheduling is reduced, implying that more users could be served. This contradicts the optimality of the assumption, completing the proof of the first statement.
    
    Now, we prove the second statement. Suppose in the optimal user scheduling, $\Pi_{a} = \left\{{\bf{X}}_{n_{1}},\dots,{\bf{X}}_{n},\dots,{\bf{X}}_{n_{2}}\right\}$ consists exclusively of batches loading model $i$, and users in ${\bf{X}}_{n}$ are not the $\left|{\bf{X}}_{n}\right|$ users with the lowest $p_{k}$ in $\mathcal{K}_{i}\setminus\bigcup\limits_{n'=n_{1}}^{n-1}\dot{\mathcal{K}}_{n'}$, where $\dot{\mathcal{K}}_{n'}$ is the set of users scheduled in batch $n'$. Then, we can adjust ${\bf{X}}_{n}$ by selecting the $\left|{\bf{X}}_{n}\right|$ users with the lowest $p_{k}$ from $\mathcal{K}_{i}\setminus\bigcup\limits_{n'=n_{1}}^{n-1}\dot{\mathcal{K}}_{n'}$. This adjustment maintains $t_{n}^{\text{load}}$ and $t_{n}^{\text{comp}}$ unchanged, as $t_{n}^{\text{load}}$ is 0, and $t_{n}^{\text{comp}}$ depends only on the batch size $\left|{\bf{X}}_{n}\right|$. However, this adjustment reduces $t_{n}^{\text{up}}$. Consequently, more users could be served in batch $n$ in the new user scheduling, which contradicts the optimality of the assumption. This completes the proof of the second statement and concludes the proof.
\section{Proof of Theorem \ref{proposition_2_2}}\label{proof:proposition_2_2}
    Suppose in the optimal user scheduling, ${\bf{X}}_{n}$ and ${\bf{X}}_{n+1}$ are non-empty, and users in ${\bf{X}}_{n}$ and ${\bf{X}}_{n+1}$ request the same model $i$. Additionally, assume that $A\left(\left|{\bf{X}}_{n}\right|+1\right)\le Q$, which indicates that batch $n$ can accommodate more than $\left|{\bf{X}}_{n}\right|$ users without exceeding the GPU memory capacity. Based on these assumptions, we can move some users from batch $n+1$ to batch $n$ within the memory constraint, resulting in ${\bf{X}}_{n}$ and ${\bf{X}}_{n+1}$ being updated to ${\bf{X}}'_{n}$ and ${\bf{X}}'_{n+1}$. Denote the data uploading time of the users in ${\bf{X}}'_{n}$ and ${\bf{X}}'_{n+1}$ by ${t'}_{n}^{\text{up}}$ and ${t'}_{n+1}^{\text{up}}$, respectively. Similarly, denote the total latency of the users in ${\bf{X}}'_{n}$ and ${\bf{X}}'_{n+1}$ by ${t'}_{n}$ and ${t'}_{n+1}$, respectively. Since ${\bf{X}}'_{n}\cup{\bf{X}}'_{n+1}$ schedules the same users as ${\bf{X}}_{n}\cup{\bf{X}}_{n+1}$, we can derive that ${t'}_{n}^{\text{up}}+{t'}_{n+1}^{\text{up}} = t_{n}^{\text{up}}+t_{n+1}^{\text{up}}$ and $\mu_i\sum\limits_{x_{n,k}\in {\bf{X}}_n\cup{\bf{X}}_{n+1} }x_{n,k} = \mu_i\sum\limits_{x_{n,k}\in {\bf{X}}'_{n}\cup{\bf{X}}'_{n+1} }x_{n,k}$. If ${\bf{X}}'_{n+1}$ becomes empty after the adjustment, then batch $n+1$ is eliminated, and the reduction in total latency is: $t_{n} + t_{n+1} - {t'}_{n} - {t'}_{n+1}=\beta_i$. Otherwise, if ${\bf{X}}'_{n+1}$ is not empty after the adjustment, then the total latency remains unchanged, satisfying: $t_{n} + t_{n+1} - {t'}_{n} - {t'}_{n+1}=0$. Therefore, in either case, the total latency follows that $t_{n+1} + t_{n} \ge {t'}_{n+1} + {t'}_{n}$. The inequality indicates that more users could be served when some users in batch $n+1$ are moved to batch $n$ in the original user scheduling under the given assumption, which contradicts the optimality of the assumption. This completes the proof.

\section{Proof of Theorem \ref{proposition_3}}\label{proof:proposition_3}
    Suppose the optimal user scheduling is $\left\{\dots,\Pi_{a},\Pi_{a+1},\Pi_{a+2},\Pi_{a+3},\dots\right\}$, where all users in $\Pi_{a}$ and $\Pi_{a+2}$ request AI models from the same model cluster $m$, users in $\Pi_{a+3}$ request models from model cluster $m_{3}$, and users in $\Pi_{a+1}=\left\{{\bf{X}}_{n_{1}},\dots, {\bf{X}}_{n_{2}}\right\}$ request models from at least one model cluster excluding cluster $m$. Additionally, the users in batch $n_{1}$ and $n_{2}$ request models from clusters $m_{1}$ and $m_{2}$, respectively, where $m_{3}\ne m$, $m_{1}\ne m$, and $m_{2}\ne m$. 
    Now, consider exchanging $\Pi_{a+1}$ with $\Pi_{a+2}$ in the original user scheduling. The completion time of $\Pi_{a+3}$ in the new user scheduling is denoted by $c'\left(\Pi_{a+3}\right)$. $c'\left(\Pi_{a+3}\right)$ follows that $c'\left(\Pi_{a+3}\right)\le c\left(\Pi_{a+3}\right)$. This is because the model loaded in the last batch of $\Pi_{a}$ and the model loaded in the first batch of $\Pi_{a+2}$ belong to the same model cluster, allowing backbone sharing across these two models. Such sharing could reduce the model loading time in the new user scheduling. Additionally, if $m_{2}=m_{3}$, $c'\left(\Pi_{a+3}\right)$ can be reduced further. Therefore, more users could be served, if $\Pi_{a+1}$ is exchanged with $\Pi_{a+2}$ in the original user scheduling $\left\{\dots,\Pi_{a},\Pi_{a+1},\Pi_{a+2},\Pi_{a+3},\dots\right\}$, which contradicts the optimality of the assumption. This completes the proof. 

\section{Proof of Corollary \ref{proposition_3_3}}\label{proof:proposition_3_3}
The first statement comes from the fact that the loading order of clusters does not impact loading time, as there is no backbone sharing between models from different clusters. 

The proof of the second statement is as follows. 
    Consider $\Pi=\left\{\dots,\Pi_{a},\Pi_{a+1},\Pi_{a+2},\Pi_{a+3},\dots\right\}$, where users in $\Pi$ request models from the same model cluster $m$, and let $\mathcal{I}'_{m}$, with $I'_{m}$ models, be the set of all requested models in $\Pi$. Based on Theorem \ref{proposition_2}, users requesting the same model are batched in the same or consecutive batches. Thus, the total data size of layers to be loaded for $\Pi$ is $\sum\limits_{i=1}^{I'_{m}}\sum\limits_{l=1}^{L_{i}}S'_{i,l}-\dot{S}_{m}$, where $\sum\limits_{i=1}^{I'_{m}}\sum\limits_{l=1}^{L_{i}}S'_{i,l}$ is the total data size of all models in $\mathcal{I}'_{m}$, with $S'_{i,l}$ being the data size of layer $l$ of model $i$, and $\dot{S}_{m}$ is the total data size of layers not requiring loading due to the backbone sharing. Since the first term is constant and independent of the model loading order, a model loading order for the users in $\Pi$ that maximizes $\dot{S}_{m}$ minimizes the total model loading time, thereby preserving the optimality of the user scheduling solution. Therefore, in the rest of the proof, we aim to determine a model loading order for the users in $\Pi$ that maximizes $\dot{S}_{m}$.
    
    To analyze the impact of model loading order on $\dot{S}_{m}$, consider exchanging two neighboring sub-sequences, say $\Pi_{a+1}$ and $\Pi_{a+2}$, in $\Pi$. This exchange affects the total data size of layers not requiring loading in $\Pi_{a+1}$, $\Pi_{a+2}$, and $\Pi_{a+3}$. Suppose that users in $\Pi_{a}$, $\Pi_{a+1}$, $\Pi_{a+2}$, and $\Pi_{a+3}$ request models $i$, $i_{1}$, $i_{2}$, and $i_{3}$, respectively, with $l_{i_{1}}\le l_{i_{2}}$ (The conclusion of the following derivation also holds when $l_{i_{1}}\ge l_{i_{2}}$, with the proof omitted.). The total data size of layers not requiring loading in $\Pi_{a+1}$, $\Pi_{a+2}$, and $\Pi_{a+3}$ before the exchange is $\dot{S}_{m,1}=\sum\limits_{l=0}^{\min\left\{l_{i},l_{i_{1}}\right\}}\tilde{S}_{m,l}+\sum\limits_{l=0}^{l_{i_{1}}}\tilde{S}_{m,l}+\sum\limits_{l=0}^{\min\left\{l_{i_{2}}, l_{i_{3}}\right\}}\tilde{S}_{m,l}$, where $\tilde{S}_{m,l}$ is the data size of the $l$-th shared bottom layer of models in model cluster $m$, with $l_{0} = 0$ and $\tilde{S}_{m,0} = 0$. Here, the three terms in $\dot{S}_{m,1}$ respectively correspond to the data size of layers not requiring loading in $\Pi_{a+1}$, $\Pi_{a+2}$, and $\Pi_{a+3}$. After exchanging $\Pi_{a+1}$ and $\Pi_{a+2}$ in $\Pi$, the new user scheduling becomes $\Pi'=\left\{\dots,\Pi_{a},\Pi_{a+2},\Pi_{a+1},\Pi_{a+3},\dots\right\}$, and the total data size of layers not requiring loading in $\Pi_{a+2}$, $\Pi_{a+1}$, and $\Pi_{a+3}$ is $\dot{S}_{m,2}=\sum\limits_{l=0}^{\min\left\{l_{i},l_{i_{2}}\right\}}\tilde{S}_{m,l}+\sum\limits_{l=0}^{l_{i_{1}}}\tilde{S}_{m,l}+\sum\limits_{l=0}^{\min\left\{l_{i_{1}}, l_{i_{3}}\right\}}\tilde{S}_{m,l}$. Then, the difference between $\dot{S}_{m,1}$ and $\dot{S}_{m,2}$ is given by $\dot{S}_{m,1}-\dot{S}_{m,2}=\sum\limits_{l=0}^{\min\left\{l_{i},l_{i_{1}}\right\}}\tilde{S}_{m,l}+\sum\limits_{l=0}^{\min\left\{l_{i_{2}}, l_{i_{3}}\right\}}\tilde{S}_{m,l}-\sum\limits_{l=0}^{\min\left\{l_{i},l_{i_{2}}\right\}}\tilde{S}_{m,l}-\sum\limits_{l=0}^{\min\left\{l_{i_{1}}, l_{i_{3}}\right\}}\tilde{S}_{m,l}$. Recalling that $l_{i_{1}}\le l_{i_{2}}$ in the assumption, we derive the following case analysis for this difference. If $l_{i}\le l_{i_{3}}$, we analyze three possible sub-cases. In the first sub-case, when $l_{i}\le l_{i_{1}}$, we have $\dot{S}_{m,1}-\dot{S}_{m,2}=\sum\limits_{l=0}^{l_{i}}\tilde{S}_{m,l}+\sum\limits_{l=0}^{\min\left\{l_{i_{2}}, l_{i_{3}}\right\}}\tilde{S}_{m,l}-\sum\limits_{l=0}^{l_{i}}\tilde{S}_{m,l}-\sum\limits_{l=0}^{\min\left\{l_{i_{1}}, l_{i_{3}}\right\}}\tilde{S}_{m,l}=\sum\limits_{l=0}^{\min\left\{l_{i_{2}}, l_{i_{3}}\right\}}\tilde{S}_{m,l}-\sum\limits_{l=0}^{\min\left\{l_{i_{1}}, l_{i_{3}}\right\}}\tilde{S}_{m,l}$. Next, we can derive that
    \begin{equation}
        \dot{S}_{m,1}-\dot{S}_{m,2}=
    \begin{cases}
    \begin{aligned}
        &\sum\limits_{l=0}^{l_{i_{3}}}\tilde{S}_{m,l}-\sum\limits_{l=0}^{l_{i_{3}}}\tilde{S}_{m,l}=0, \ \text{if}\ l_{i_{3}}\le l_{i_{1}},\\
        &\sum\limits_{l=0}^{l_{i_{3}}}\tilde{S}_{m,l}-\sum\limits_{l=0}^{l_{i_{1}}}\tilde{S}_{m,l}\ge 0,\ \text{if} \ l_{i_{1}}\le l_{i_{3}}\le l_{i_{2}},\\
        &\sum\limits_{l=0}^{l_{i_{2}}}\tilde{S}_{m,l}-\sum\limits_{l=0}^{l_{i_{1}}}\tilde{S}_{m,l}\ge 0,\ \text{if} \  l_{i_{2}}\le l_{i_{3}}.   
    \end{aligned}
    \end{cases}.
    \end{equation}
    Thus, $\dot{S}_{m,1}-\dot{S}_{m,2}\ge 0$ consistently holds in this sub-case.
    In the second sub-case, when $l_{i_{1}}\le l_{i}\le l_{i_{2}}$, we have $\dot{S}_{m,1}-\dot{S}_{m,2}=\sum\limits_{l=0}^{l_{i_{1}}}\tilde{S}_{m,l}+\sum\limits_{l=0}^{\min\left\{l_{i_{2}}, l_{i_{3}}\right\}}\tilde{S}_{m,l}-\sum\limits_{l=0}^{l_{i}}\tilde{S}_{m,l}-\sum\limits_{l=0}^{\min\left\{l_{i_{1}},l_{i_{3}}\right\}}\tilde{S}_{m,l}$. 
    Since $l_{i_{1}}\le l_{i}\le l_{i_{2}}$ holds in this sub-case and $l_{i}\le l_{i_{3}}$ holds in the case overall, we have $\sum\limits_{l=0}^{\min\left\{l_{i_{1}},l_{i_{3}}\right\}}\tilde{S}_{m,l}=\sum\limits_{l=0}^{l_{i_{1}}}\tilde{S}_{m,l}$. Furthermore, we can conclude that $\dot{S}_{m,1} - \dot{S}_{m,2} \ge 0$ in the second sub-case.
    In the third sub-case, when $l_{i_{2}} \le l_{i}$, we can derive that $\dot{S}_{m,1}-\dot{S}_{m,2}=\sum\limits_{l=0}^{l_{i_{1}}}\tilde{S}_{m,l}+\sum\limits_{l=0}^{l_{i_{2}}}\tilde{S}_{m,l}-\sum\limits_{l=0}^{l_{i_{2}}}\tilde{S}_{m,l}-\sum\limits_{l=0}^{l_{i_{1}}}\tilde{S}_{m,l}=0$ based on $l_{i}\le l_{i_{3}}$ for the overall case. Consequently, $\dot{S}_{m,1}-\dot{S}_{m,2}\ge 0$ always holds when $l_{i}\le l_{i_{3}}$. This indicates that scheduling users in $\Pi$ in ascending order of $l_{i}$ of the requested models maximizes $\dot{S}_{m}$. Similarly, $\dot{S}_{m,2}-\dot{S}_{m,1}\ge 0$ always holds if $l_{i}\ge l_{i_{3}}$. This implies that scheduling users in $\Pi$ in descending order of $l_{i}$ of the requested models maximizes $\dot{S}_{m}$. Since we set $l_{0}$ to 0, scheduling users in ascending order of $l_{i}$ of the requested models preserves the optimality of the user scheduling. 
    This completes the proof of the second statement and concludes the proof of Corollary \ref{proposition_3_3}.

\end{appendices}

\end{document}